%% file: sample-sigconf.tex
\documentclass[sigconf]{acmart}

\usepackage{amsmath}
\usepackage{amsthm}
\usepackage{subfigure}
\usepackage{graphicx}
\usepackage{enumitem}
\usepackage{multirow}
\usepackage{makecell}
\usepackage{graphicx}
\usepackage{subfloat}
\usepackage{color}
\usepackage{caption}
\usepackage{subcaption}
\usepackage{enumitem}
\usepackage{xspace}
\usepackage{xcolor} 
\usepackage{mdframed}
\usepackage[ruled]{algorithm2e}
\newcommand{\etal}{\emph{et al.}\xspace}

\newtheorem{proposition}{Proposition}[section]
\AtBeginDocument{%
  }

\setcopyright{acmcopyright}
\copyrightyear{2018}
\acmYear{2018}
\acmDOI{XXXXXXX.XXXXXXX}

\acmConference[Conference acronym 'XX]{Make sure to enter the correct
  conference title from your rights confirmation emai}{June 03--05,
  2018}{Woodstock, NY}
\acmPrice{15.00}
\acmISBN{978-1-4503-XXXX-X/18/06}

\newcommand{\method}{ScalingNote\xspace}
\newcommand{\methoddual}{ScalingDual\xspace}
\newcommand{\methoddoc}{ScalingDoc\xspace}
\begin{document}

\title{\method: Scaling up Retrievers with Large Language Models for Real-World Dense Retrieval}
\renewcommand{\thefootnote}{} 
\author{Suyuan Huang*$^1$,\footnotetext{* Equal contribution}\footnotetext{$\dag$ Corresponding authors} Chao Zhang*$^{2,3}$, Yuanyuan Wu$^4$, Haoxin Zhang$^4$, Yuan Wang$^4$, Maolin Wang$^3$, Shaosheng Cao$^{\dag4}$, Tong Xu$^{\dag2}$, Xiangyu Zhao$^{\dag3}$, Zengchang Qin$^{\dag1}$, Yan Gao$^4$, Yunhan Bai$^4$, Jun Fan$^4$, Yao Hu$^4$, Enhong Chen$^{\dag2}$}
 \affiliation{%
 \institution{$^1$Beihang University}
 \country{}
 }
 \affiliation{
 \institution{$^2$University of Science and Technology of China \& State Key Laboratory of Cognitive Intelligence}
 \country{}
 }
 \affiliation{
 \institution{$^3$City University of Hong Kong}
 \country{}
 }
\affiliation{%
 \institution{$^4$Xiaohongshu Inc.}
 \country{}
 }
 \renewcommand{\shortauthors}{Suyuan Huang, Chao Zhang \etal}

\begin{abstract}

Dense retrieval in most industries employs dual-tower architectures to retrieve query-relevant documents.
Due to online deployment requirements, existing real-world dense retrieval systems mainly enhance performance by designing negative sampling strategies, overlooking the advantages of scaling up.
Recently, Large Language Models (LLMs) have exhibited superior performance that can be leveraged for scaling up dense retrieval.
However, scaling up retrieval models significantly increases online query latency.
To address this challenge, we propose \method, a two-stage method to exploit the scaling potential of LLMs for retrieval while maintaining online query latency.
The first stage is training dual towers, both initialized from the same LLM, to unlock the potential of LLMs for dense retrieval.
Then, we distill only the query tower using mean squared error loss and cosine similarity to reduce online costs.
Through theoretical analysis and comprehensive offline and online experiments, we show the effectiveness and efficiency of \method.
Our two-stage scaling method outperforms end-to-end models and verifies the scaling law of dense retrieval with LLMs in industrial scenarios, enabling cost-effective scaling of dense retrieval systems.
Our online method incorporating \method significantly enhances the relevance between retrieved documents and queries.

\end{abstract}


\keywords{Dense Retrieval; Large Language Model; Scaling Law; Knowledge Distillation}

\maketitle

\input{chapters/01_intro.tex}

\input{chapters/03_problem_statement.tex}
\input{chapters/04_method.tex}

\input{chapters/05_exp.tex}

\input{chapters/02_related.tex}

\input{chapters/06_conclusion.tex}

\clearpage
\bibliographystyle{ACM-Reference-Format}
\bibliography{sample-base}

\appendix
\input{chapters/07_appendix.tex}

\end{document}

%% file: chapters/01_intro.tex
\section{Introduction}
\begin{figure}[!t]
    \centering
    \includegraphics[width=0.45\textwidth]{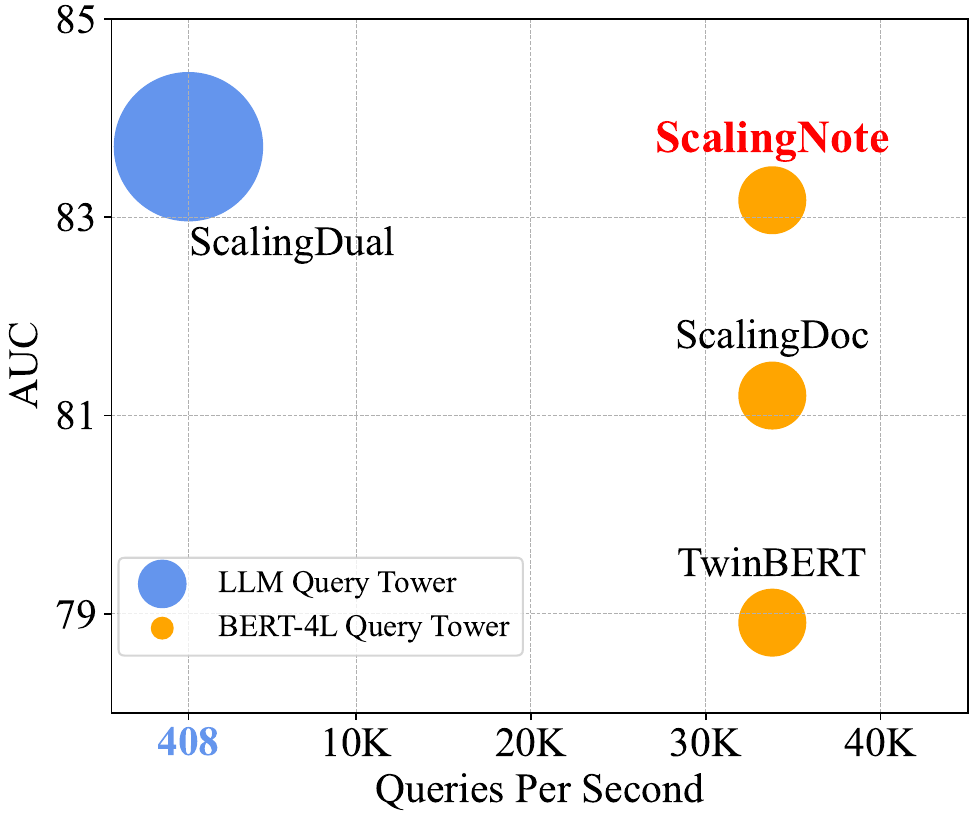}
    \caption{Performance and queries per second (QPS) comparison of scaling strategies for dual-tower dense retrieval. The circle sizes mean the number of parameters in query towers. Our \method gets high AUC without reducing QPS.}
    \label{fig:compare}
\end{figure}
\renewcommand{\thefootnote}{\arabic{footnote}}
Information retrieval platforms, such as Google\footnote{\url{www.google.com}}, Baidu\footnote{\url{www.baidu.com}}, and Xiaohongshu\footnote{\url{www.xiaohongshu.com}}, are important Internet applications that enable users to access relevant data from the vast amount of information available online~\cite{zhang2022uni,liu2021pre,huang2020embedding}.
Dense retrieval is a classic techniques in retrieval~\cite{karpukhin2020dense,ni2022large}. 
This method employs dual-tower architectures to encode both queries and documents into dense embeddings.
The semantic relevance between queries and documents is measured using embedding similarity.
With the rapid advancement of pre-training techniques, pre-trained models like BERT~\cite{kenton2019bert} and GPT~\cite{brown2020language} have greatly enhanced the semantic understanding capabilities of dense retrieval models, making dense retrieval a primary channel in industrial retrieval~\cite{huang2020embedding,liu2021que2search,he2023que2engage}.

In real-world retrieval systems, the massive number of real-time queries, combined with strict latency requirements, often limits the size of retrieval models to ensure quick responses~\cite{liu2021pre,he2023que2engage,magnani2022semantic}.
To prevent burdening deployment, there is a greater focus on refining negative sampling techniques~\cite{hofstatter2021efficiently,xu2022negative,xiongapproximate} to enhance model performance, while the potential advantages of scaling up model size are frequently underexplored.
Recently, decoder-only Large Language Models (LLMs) have demonstrated huge potential in natural language understanding and generation~\cite{zhao2023survey,xu2023large}. 
This advancement has also impacted dense textual retrieval~\cite{ma2024fine,liao2024d2llm,zhu2023large,behnamghader2024llm2vec,li2024llama2vec}.
Concurrently, scaling laws based on bi-directional encoders for dense textual retrieval have been verified~\cite{fang2024scaling}, indicating that more parameters and more data lead to more powerful retrieval models.
This all suggests a promising direction for scaling up dense retrieval with larger-scale decoder-only LLMs in real industry scenarios.

However, scaling up retrievers in industrial scenarios faces more challenges than existing research on scaling dense retrieval.
In industrial applications, different towers of retrievers serve different purposes.
Query towers process real-time queries from users, prioritizing low latency.
They typically have fewer parameters~\cite{liu2021que2search,kimustad,magnani2022semantic} or use high-level quantization~\cite{liu2021pre,huang2020embedding} to speed up online inference.
On the other hand, document towers generate embeddings for existing documents to update search indexes.
Since they operate offline, document towers can afford more parameters~\cite{liu2021que2search,he2023que2engage,magnani2022semantic} and complex computations~\cite{zhang2024notellm,zhang2024notellm2}.
However, storage of extensive document embeddings~\cite{lu2024knowledge} and the memory cost of caching them~\cite{liu2021pre} is a primary consideration.

In this work, we explore cost-effective methods for scaling up retrievers with LLMs for real-world dense textual retrieval.
Based on dual-tower architectures, we investigate two approaches to achieve retriever scaling.
The first approach, \methoddual, is scaling both the query and document towers simultaneously. 
This method fully leverages the potential of LLMs for dense retrieval. 
However, it is challenging to satisfy online query latency requirements.
Thus, we propose a second scaling approach, \methoddoc.
This method is scaling only the document tower with LLMs while maintaining the query tower's original structure. 
By enhancing document representation without compromising online query latency, this approach is better suited for online deployment.
In Figure~\ref{fig:compare}, we compare these two scaling methods using the same LLM, Qwen 2.5 7B~\cite{qwen2.5}.
We observe that although \methoddoc outperforms TwinBERT~\cite{lu2020twinbert} significantly, it still has a huge performance gap compared to \methoddual.
On the other hand, \methoddual lags behind in terms of Queries Per Second (QPS) and requires more memory to cache the query tower parameters.

For fully \textbf{Scaling} real-world dense retrieval and meeting online query latency requirements for \textbf{Note} scenarios in Xiaohongshu, we propose \method, a two-stage training method.
In the first stage, our framework trains dual towers, both initialized from the same LLM.
The goal of this stage is to fully unlock the potential of LLMs for dense textual retrieval.
Specifically, we utilize contrastive learning to recognize truly relevant documents from cross-device in-batch negatives. 
To further enhance the performance of LLM-based retrievers, we mine hard negatives using margin loss.
In the second stage, we propose Query-based Knowledge Distillation (QKD), which distills the knowledge from the query tower to an online query encoder based on queries. 
We utilize mean squared error (MSE)~\cite{wang2023query} and cosine similarity as loss functions.
We have also theoretically proven that \method can achieve a lower upper bound on generalization error compared to \methoddoc with the same size of online query tower and document tower.
Besides, we describe the workflow of our online system after integrating our model.
During the offline document index construction phase, each new document is encoded by our LLM-based document tower and generates semantic IDs according to the residual K-means~\cite{macqueen1967some} clustering distance.
The embeddings of semantic IDs are combined with additional features to generate the final document embeddings for constructing the IVFPQ index~\cite{jegou2010product}.
During the online query retrieval phase, each query is encoded by our query tower in real-time, and the system integrates user historical information to conduct a personalized search based on IVFPQ.
We conduct offline and online experiments to demonstrate the effectiveness and efficiency of \method and verify the scaling law of dense retrieval with decoder-only LLMs in industrial scenarios.

Our main contributions are summarized as follows:
\begin{itemize}[leftmargin=*]
    \item We scale up dense textual retrieval using LLMs for industrial scenarios and verify the scaling laws of dense retrieval.
    \item 
    We report lessons learned from scaling up dense retrieval for industrial applications.
    Based on these, we propose an effective two-stage scaling framework, \method. 
    We also prove our method's effectiveness theoretically.
    \item We conduct comprehensive experiments, both online and offline, to evaluate our approach.
    Our model achieves significant improvement and maintains low online latency.
\end{itemize}

%% file: chapters/03_problem_statement.tex
\section{Problem Definition}
In a real-world retrieval system, there exists a large corpus of documents denoted as $\mathbb{C}=\{d_i\}_{i=1}^{|\mathbb{C}|}$, where 
$|\mathbb{C}|$ is typically in the billions.
Each documents $d_i$ is represented as a tuple $d_i=(t_i, tp_i, c_i)$, where $t_i$ is the title, $tp_i$ is the topic, and $c_i$ is the content.
The system receives the online query $q$ from the user.
Dense retrieval utilizes a relevance score to measure the similarity between a document $d$ and a given query $q$.
It employs a dual-tower architecture: a query tower $Q$ that encodes the given query $q$ into the query embedding $\boldsymbol{q}$, and a document tower $D$ that encodes the given document $d$ into the document embedding $\boldsymbol{d}$.
The relevance score $s$ is calculated as $s = sim(\boldsymbol{q},\boldsymbol{d})$, where $ sim(\boldsymbol{q},\boldsymbol{d})$ is typically the cosine similarity or dot product between the query and document embeddings.
Dense retrieval systems aim to identify and return the most query-relevant document subset $\mathbb{C}^q$ from the complete corpus $\mathbb{C}$ by ranking documents according to their relevance scores with the query $q$.

%% file: chapters/04_method.tex
\section{Methodology}

\begin{figure*}[htbp]
    \centering
    \includegraphics[width=\textwidth]{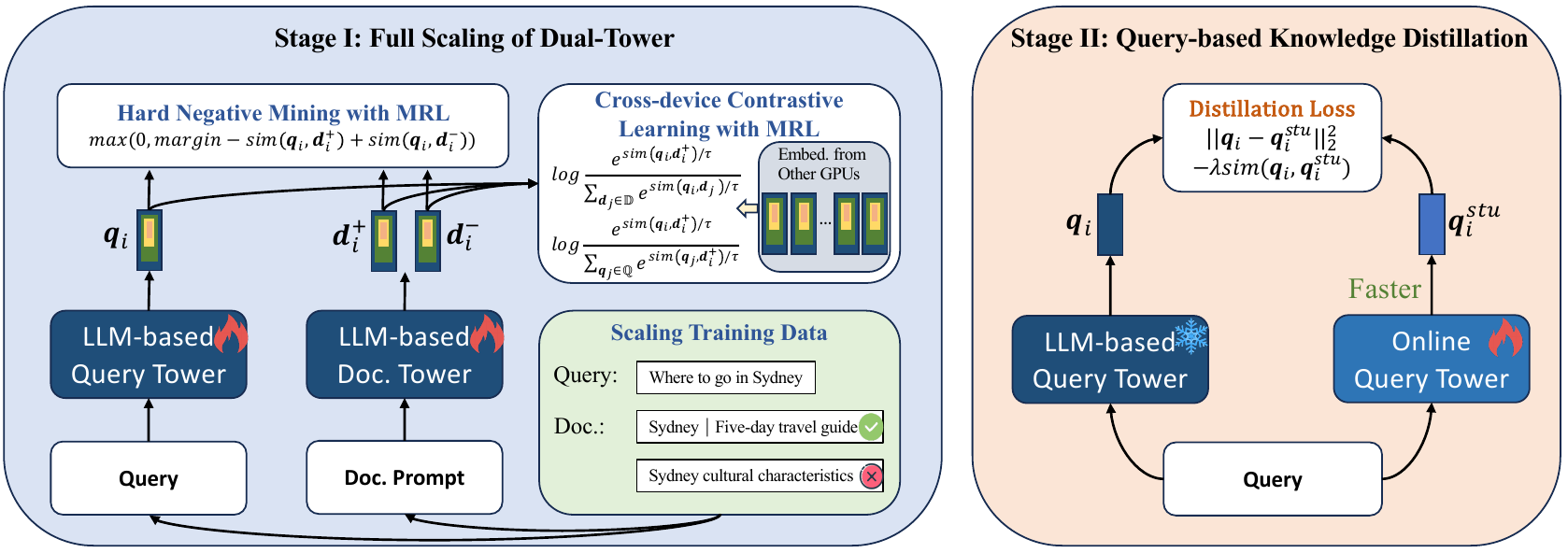}
    \caption{The framework of \method. The first stage is fully scaling the dual-tower using scaled training data, which learns through cross-device contrastive learning and hard negative mining.
    The second stage is query-based knowledge distillation (QKD), which transfers the scaled query knowledge from the LLM-based query tower to the faster online query tower.}
    \label{fig:framework}
\end{figure*}

\subsection{Overview of \method}
We present the overall pipeline of \method in Figure~\ref{fig:framework}.
Our approach consists of two stages.
The first stage maximizes the LLM's potential for dense retrieval.
We employ contrastive learning and hard negative mining techniques to enhance the LLM's retrieval capabilities.
The second stage optimizes online query efficiency.
We implement QKD to transfer query-related knowledge from the LLM-based query tower to a more efficient BERT-based query tower.
This two-stage approach allows \method to maintain the high precision offered by scaling while achieving online query efficiency comparable to existing methods.

\subsection{Stage I: Full Scaling of Dual-Tower}

\textbf{Prompt Design.} To leverage the textual understanding capabilities of LLMs, we need to design task-specific prompts.
While existing works often use compression prompts to condense all information into a single token via generation~\cite{zhang2024notellm,jiang2023scaling}, this approach faces two main challenges in our scenarios.
Firstly, documents often contain extensive textual information, including titles, topics, and content. 
Compressing all information into a single token is challenging and may lead to loss of details. 
This is particularly problematic in Xiaohongshu, where documents are represented primarily by titles in the document display interface.
If the compressed tokens focus mainly on content due to relative distance within the document, it may result in the platform showing less title-relevant documents, potentially reducing user engagement.
Secondly, relying solely on compressing information into a single token may not fully utilize LLMs.
LLMs' capabilities can be enhanced by using more tokens for more comprehensive information processing~\cite{pfau2024let,zhang2024simple}.

To solve the above challenges, we design the prompt as follows:
\begin{mdframed}
    \textbf{Document Prompt for Dense Retrieval} \\
    Document: \{'title': $t_i$\}. Predict a query term:"[TITLE\_QUERY]". Document: \{'topic': $tp_i$, 'content': $c_i$\}. Predict a query term: \\
    "[CONTENT\_QUERY]", combine the predicted query terms, and compress the above content into one word:"[EMB]".
\end{mdframed}
In this prompt, [TITLE\_QUERY], [CONTENT\_QUERY] and [EMB] are special placeholders. 
To address the issue of models potentially ignoring titles due to their distance from the compressed token, we set the LLM to predict the title query terms independently.
The predicted pseudo query terms from the title are then used for further optimization.
Besides, to utilize the reasoning and correction ability of LLMs, we propose a inductive prompt~\cite{pfau2024let,zhang2024simple} to guide the LLM to summarize the predicted query terms, and generate a higher-quality new token, represented by [EMB].

Due to the online cost of query towers, we propose making the prompt as short as possible.
Following~\cite{ma2024fine,wang2023improving}, we append the [EOS] token, which means the end of a sentence, directly after the queries.

\textbf{Encoding.} 
After inputting the real value into the prompt template, we tokenize $d_i$ into the discrete token list $T^{d_i}=[t^{d_i}_1,...,t^{d_i}_{|T^{d_i}|}]$ and feed the textual information into the document tower $D$ to get the last hidden states $\boldsymbol{H}^{d_i} \in \mathbb{R}^{|T^{d_i}|\times h} $:
\begin{equation}
\boldsymbol{H}^{d_i} = D([t^{d_i}_1,...,t^{d_i}_{|T^{d_i}|}]),
\end{equation}
where $h$ is the hidden state size of LLM $D$.
We take the previous last hidden state of the special placeholders [TITLE\_QUERY], [CONTENT\_QUERY], and [EMB] as $\boldsymbol{h}^{d_i}_{t} \in \mathbb{R}^h$, $\boldsymbol{h}^{d_i}_{c} \in \mathbb{R}^h$, and $\boldsymbol{h}^{d_i}_{e} \in \mathbb{R}^h$, respectively.
Similarly, the query $q_i$ is tokenized as $T^{q_i}=[t^{q_i}_1,...,t^{q_i}_{|T^{q_i}|}]$.
Then, the query tower $Q$ encodes the query $q_i$ as follows:
\begin{equation}
\boldsymbol{H}^{q_i} = Q([t^{q_i}_1,...,t^{q_i}_{|T^{q_i}|}]),
\end{equation}
where $\boldsymbol{H}^{q_i} \in \mathbb{R}^{|T^{q_i}|\times h} $, and we take the last position hidden state $\boldsymbol{h}^{q_i} \in \mathbb{R}^h$ as the query embedding.

To facilitate the future use of our representations by downstream applications with varying resource constraints, we adopt efficient MRL~\cite{kusupati2022matryoshka}.
We set the target vector dimension as $dim$, and the least dimension as $dim_{low}$.
Then, we obtain the candidate available vector dimensions $\mathbb{M} = \{dim_{low}, 2dim_{low},...,\frac{dim}{2}, dim\}$, where $|\mathbb{M}| \leq \left\lfloor \log(dim) \right\rfloor$ and $dim \ll h$.
We set the transform matrix for the document tower as $\boldsymbol{W}^D \in \mathbb{R}^{dim\times h}$, and the transform matrix for the query tower as $\boldsymbol{W}^Q \in \mathbb{R}^{dim\times h}$.
Therefore, we can get the maximum downstream embeddings:
\begin{equation}
\boldsymbol{d}^{t}_i = \boldsymbol{W}^D \boldsymbol{h}^{d_i}_t, \quad \boldsymbol{d}^{c}_i = \boldsymbol{W}^D \boldsymbol{h}^{d_i}_c, \quad \boldsymbol{d}^{e}_i = \boldsymbol{W}^D \boldsymbol{h}^{d_i}_e,
\end{equation}
\begin{equation}
\boldsymbol{q}_i = \boldsymbol{W}^Q \boldsymbol{h}^{q_i}.
\end{equation}
These embeddings are transformed into lower dimensions to reduce memory cost.
We can select any dimension $m \in \mathbb{M}$ and truncate these embeddings to form shorter representations, such as ${\boldsymbol{d}_{i}^{t}}{\scriptstyle[1:m]}$.

\textbf{Contrastive Learning.} 
We assume the presence of $K$ GPU devices, with each device processing a batch containing $B$ data triplets denoted as $\{\langle q_i,d_i^+,d_i^-\rangle\}_{i=1}^{B}$, where $q_i$ represents the query, $d_i^+$ is the positive document, and $d_i^-$ is the hard negative document. 
We will introduce the details of data construction in Section~\ref{data_construction}.
Upon completion of encoding by each GPU device, we utilize cross-device contrastive learning to expand the number of in-batch negatives.
Through the all-gather operation, each device obtains all query and document representations from other devices.
We denote the batch of all-gathered query representations as $\mathbb{Q}$, and and the batch of all-gathered document representations of three types as $\mathbb{D}^t$, $\mathbb{D}^c$, and $\mathbb{D}^e$.
For any $m\in \mathbb{M}$, we conduct contrastive learning for queries and three types of documents embeddings.
We take the contrastive learning for queries and document title as an example:
\begin{equation}
\begin{aligned}
L^{t,m}_{q2d}&=- \frac{1}{|\mathbb{Q}|} \sum_{i=1}^{|\mathbb{Q}|} log \frac{e^{sim(\boldsymbol{q}_i{\scriptstyle[1:m]}, \boldsymbol{d}_i^{t+}{\scriptstyle[1:m]})/\tau}}{\sum_{\boldsymbol{d}_j^{t}\in \mathbb{D}^t}e^{sim(\boldsymbol{q}_i{\scriptstyle[1:m]}, \boldsymbol{d}_j^{t}{\scriptstyle[1:m]})/\tau}},
\end{aligned}
\end{equation}
\begin{equation}
\begin{aligned}
L^{t,m}_{d2q}&=- \frac{1}{|\mathbb{D}^t|} \sum_{i=1}^{|\mathbb{D}^t|} log \frac{e^{sim(\boldsymbol{q}_i{\scriptstyle[1:m]}, \boldsymbol{d}_i^{t+}{\scriptstyle[1:m]})/\tau}}{\sum_{\boldsymbol{q}_j\in \mathbb{Q}}e^{sim(\boldsymbol{q}_j{\scriptstyle[1:m]}, \boldsymbol{d}_i^{t+}{\scriptstyle[1:m]})/\tau}},
\end{aligned}
\end{equation}
where $L^{t,m}_{q2d}$ is the loss for selecting the corresponding document title representation given a query, and $L^{t,m}_{d2q}$ is the loss for selecting the corresponding query representation given a document title embedding.
$sim()$ function is cosine similarity.
And $\tau$ is the temperature hyperparameter.
We then aggregate these two types of losses across different embedding dimensions $m$ as follows:
\begin{equation}
L^{t} = \frac{1}{2} \sum_{m\in \mathbb{M}} w^m (L^{t,m}_{q2d}+L^{t,m}_{d2q}),
\end{equation}
where $w^m$ are the hyperparameters to control the loss intensity of different dimensions.
Similarly, for the other two types of document representations, we can obtain $L^c$ and $L^e$.
Finally, we can calculate the overall contrastive learning loss:
\begin{equation}
L^{con} = L^{t} + L^{c} + L^{e}.
\end{equation}

\textbf{Hard Negative Mining.} 
Hard negatives are crucial for model performance, especially for top-position ranking~\cite{xiongapproximate,zhan2021optimizing}.
Therefore, we mine the hard negatives using margin loss~\cite{huang2020embedding}.
We take the loss for document title embeddings as an example.
For any $m\in \mathbb{M}$, we calculate the following margin loss:
\begin{equation}
\begin{aligned}
L^{t,m}_{hard} = \frac{1}{|\mathbb{Q}|} \sum_{i=1}^{|\mathbb{Q}|} max(0, & margin - sim(\boldsymbol{q}_i{\scriptstyle[1:m]}, \boldsymbol{d}_i^{t+}{\scriptstyle[1:m]}) \\
& + sim(\boldsymbol{q}_i{\scriptstyle[1:m]}, \boldsymbol{d}_i^{t-}{\scriptstyle[1:m]})),
\end{aligned}
\end{equation}
where $margin$ is a hyperparameter.
Next, we aggregate the hard margin loss as follows:
\begin{equation}
L^{t}_{hard} = \sum_{m\in \mathbb{M}} w^m_{hard} L^{t,m}_{hard},
\end{equation}
where $w^m_{hard}$ is the weight for each dimension $m$ and $L^{t}_{hard}$ is the loss for hard negatives related to the document title.
Similarly, we can also obtain the losses $L^{c}_{hard}$ and $L^{e}_{hard}$ for the other two types of document embeddings.
The overall hard negative margin loss is computed as follows:
\begin{equation}
L^{hard} = L^{t}_{hard} + L^{c}_{hard} + L^{e}_{hard}.
\end{equation}

The final training loss is a weighted sum:
\begin{equation}
L= L^{con} + \alpha L^{hard},
\end{equation}
where $\alpha$ is the hyperparameter to control the intensity of the impact of hard negatives.
We optimize the final loss to update both the query tower and document tower.

\subsection{Stage II: Query-based Knowledge Distillation}
In the first stage, we simultaneously optimize the query tower and the document tower.
However, the LLM-based query tower significantly impacts online query latency.
Therefore, it is necessary to reduce the online inference time by minimizing the model size.

Compared to documents, queries are shorter and contain less information. 
This makes knowledge transfer based on queries easier and more efficient.
Therefore, instead of using relevance scores from teacher models to guide student models in understanding the relationship between queries and documents~\cite{lin2021batch,kimustad}, which still requires loading the LLM-based document tower, we leverage embedding-based distillation~\cite{wang2023query,campos2023quick,kimustad}.
We adopt QKD which decouples the document tower and focuses solely on learning query representations from the query tower.
This approach only utilizes the query input.
We denote the student model of the query tower as $Q_{stu}$, which can have totally different model architecture and vocabularies.
The vocabulary can be augmented with frequently occurring new tokens specific to our platform.
We insert [CLS] token before the query $q_i$ and then tokenize it using the tokenizer of the student model as $T^{q_i}_{stu}=[t^{q_i}_1,...,t^{q_i}_{|T^{q_i}_{stu}|}]$.
Next, we encode the tokens using the student model:
\begin{equation}
\boldsymbol{H}^{q_i}_{stu} = Q_{stu}([t^{q_i}_1,...,t^{q_i}_{|T^{q_i}_{stu}|}]).
\end{equation}
We take the last hidden state of the [CLS] token as the query embedding from the student model $\boldsymbol{h}^{q_i}_{stu}\in \mathbb{R}^{h'}$.
$h'$ means the hidden state dimension of the student model $Q_{stu}$.
We utilize a matrix $\boldsymbol{W}^{Q_{stu}}\in \mathbb{R}^{dim\times h'}$ to reduce the dimension of student query embedding to the target dimension:
\begin{equation}
\boldsymbol{q}_i^{stu} = \boldsymbol{W}^Q_{stu} \boldsymbol{h}^{q_i}.
\end{equation}
In the first stage, our model measures distance using cosine similarity, which effectively captures the directional alignment of embeddings but overlooks the importance of magnitude.
To address this limitation, we also use MSE~\cite{wang2023query} to align the representations of the teacher and student models more comprehensively, considering both direction and magnitude.
The QKD loss is defined as:
\begin{equation}
L_{KD} = \frac{1}{|\mathbb{Q}|} \sum_{i=1}^{|\mathbb{Q}|} (  || \boldsymbol{q}_i-\boldsymbol{q}_i^{stu} ||_2^2 - \lambda sim(\boldsymbol{q}_i,\boldsymbol{q}_i^{stu})),
\end{equation}
where $\lambda$ is a hyperparameter.
After training in the second stage, the student query tower can closely mimic the teacher query tower's representations, so that it is compatible with the LLM-based document tower.

\subsection{Theoretical Analysis}
To further explain the effectiveness of \method, we provide a theoretical analysis comparing the generalization upper bounds of \methoddoc and \method when both of them have the same query and document parameter sizes.
We extend the analysis in \cite{kimustad}.
See Appendix~\ref{theory} for the proof.

First, we discuss the generalization bound of \methoddoc.
\begin{proposition}
\label{pro1}
Let $y\in \{0, 1\}$ be the label indicating the relevance of a query-document pair, where $1$ denotes relevance and $0$ denotes irrelevance.
Assume $\mathbb{S}_n=\{(q_i,d_i,y_i)\}_{i=1}^{n}$ is the training dataset.
Let $\ell$ be the overall loss, which is $L_{\ell}$-Lipschitz in its non-target variable.
Denote $\mathcal{Q}$ and $\mathcal{D}$ as the function classes for the query tower $Q_{D}$ and the document tower $D_{D}$, respectively. 
Define $s_{q,d}^{Q_{D},D_{D}}$ as the relevance score generated by the query tower $Q_{D}$ and the document tower $D_{D}$ for query $q$ and document $d$. 
The $L_2$-norm of the query and document embeddings is bounded by $K$ for each query and document tower in our function class.
Then,
\begin{equation}
\begin{aligned}
    \mathbb{E}_{q,d}\ell(s_{q,d}^{Q_{D},D_{D}}, y) &\leq R(s_{q,d}^{Q_{D},D_{D}};\mathbb{S}_n)  \\
    & +  \underbrace{\mathbb{E}_{\mathbb{S}_n} \frac{48KL_{\ell}}{\sqrt{n}} \int_0^{\infty} \sqrt{\log N(u, \mathcal{Q}) + \log N(u, \mathcal{D})}\ du }_{\text{Denoted as } U_{Q_{D},D_{D}}}.
\end{aligned}
\label{e2e_eq}
\end{equation}
$N(u, \cdot)$ is the $u$-covering number of a function class.
\end{proposition}
In Eq.~\ref{e2e_eq}, $\mathbb{E}_{q,d}\ell(s_{q,d}^{Q_{D},D_{D}}, y)$ is the expected risk, which shows the generalization ability of dense retrival, and $R(s_{q,d}^{Q_{D},D_{D}};\mathbb{S}_n)$ is the emperical risk, which is the error in the training dataset.

Next, we discuss the expected risk of \method.
\begin{proposition}
\label{pro2}
Let $Q_R$ denote the first stage query tower, $D_R$ denote the first stage document tower, and $Q_{stu}$ denote the student query tower with the same architecture as the \methoddoc's query tower $Q_D$.
Therefore, the function classes of the student query tower $Q_{stu}$ is also $\mathcal{Q}$.
The student query tower $Q_{stu}$ shares a frozen document tower $D_R$ with the large query tower $Q_R$.
And $D_R$ has the same architecture as \methoddoc's document tower $D_D$.
Then,
\begin{equation}
\begin{aligned}
    \mathbb{E}_{q,d}\ell(s_{q,d}^{Q_{stu},D_R}, y) &\leq R(s_{q,d}^{Q_R,D_R};\mathbb{S}_n) + \frac{2K}{n}\sum_{i\in[n]} \|\boldsymbol{q}_i^{stu}-\boldsymbol{q}_i\|  \\
    & + \underbrace{\mathbb{E}_{\mathbb{S}_n} \frac{48KL_{\ell}}{\sqrt{n}} \int_0^{\infty} \sqrt{\log N(u, \mathcal{Q})}\ du}_{\text{Denoted as } U_{Q_{stu},D_R}}.
\end{aligned}
\label{2stage_eq}
\end{equation}
$\boldsymbol{q}^{stu}_i$ and $\boldsymbol{q}_i$ are the query embeddings generated by the student query tower $Q_{stu}$ and the first stage large query tower $Q_R$ for query $q_i$, respectively.
\end{proposition}

\textbf{Analysis.} We can compare the upper bound of $\mathbb{E}_{q,d}\ell(s_{q,d}^{Q_D,D_D}, y)$ in Eq.~\ref{e2e_eq} and $\mathbb{E}_{q,d}\ell(s_{q,d}^{Q_{stu},D_R}, y)$ in Eq.~\ref{2stage_eq}.
Based on existing works~\cite{fang2024scaling,campos2023quick,ni2022large} and Figure~\ref{fig:compare}, larger and more robust models can achieve lower empirical risk on the training dataset with the same training settings.
Thus, $R(s_{q,d}^{Q_R,D_R};\mathbb{S}_n) \leq R(s_{q,d}^{Q_D,D_D};\mathbb{S}_n)$.
Besides, due to $ N(u, \mathcal{D})\geq 1$, we have $U_{Q_{stu},D_R} \leq U_{Q_D,D_D}$.
The factor affecting the upper bound of expected risk for \method is the embedding distance $\frac{2K}{n}\sum_{i\in[n]} \|\boldsymbol{q}_i^{stu}-\boldsymbol{q}_i\|$.
This distance is easily reduced to a low value.
Consequently, \method can achieve a lower upper bound of expected risk than \methoddoc.

\subsection{Online Serving}

\begin{figure}[t]
    \centering
    \includegraphics[width=0.5\textwidth]{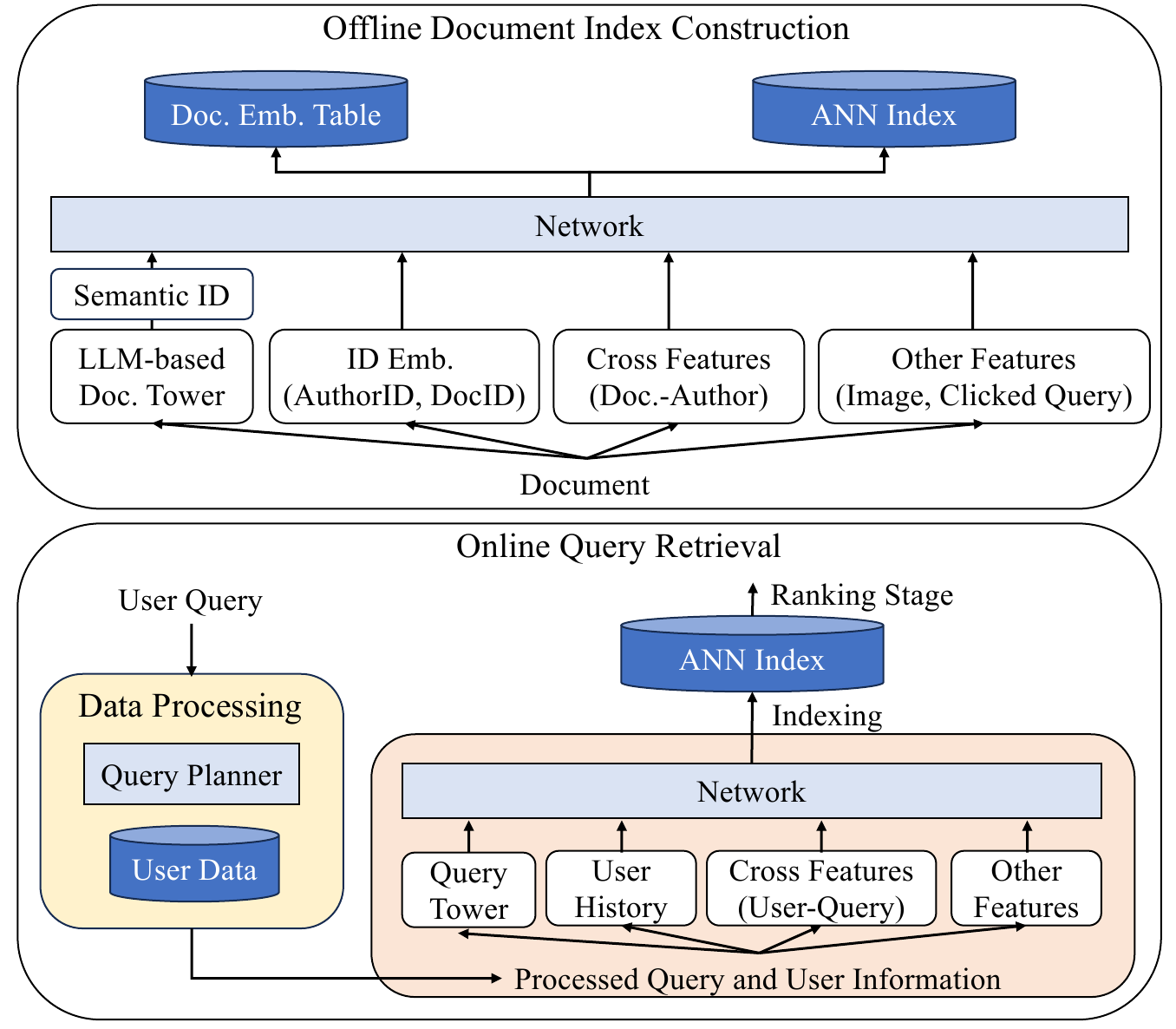}
    \caption{The online framework includes two stages: offline document index construction and online query retrieval. }
    \label{fig:online_framework}
\end{figure}

In this subsection, we introduce the online serving algorithm incorporating LLM-based dense retrieval. 
Our system operates in two stages: offline document index construction and online query retrieval, as illustrated in Figure~\ref{fig:online_framework}.
During the offline stage, our system encodes document-side information.
We first use LLM-based document tower to encode documents from a data pool. 
Then, we employ 6-layer residual K-means~\cite{macqueen1967some} clustering.
Every time the system encodes a new document, the LLM-based document tower is used to generate the embedding and cluster it into specific clustering centers. 
We assign the clustering center IDs as the semantic IDs.
We use retrained semantic ID embeddings to represent the semantic information of the document.
Besides, our whole model incorporates other document features, such as ID embeddings, cross features, images, etc.
All inputs are processed through a neural network to generate comprehensive document embeddings. 
The final embeddings are then stored in document embedding tables and used to construct ANN indexes based on IVFPQ~\cite{jegou2010product}.

\begin{figure*}[htbp]
    \centering
    \includegraphics[width=\textwidth]{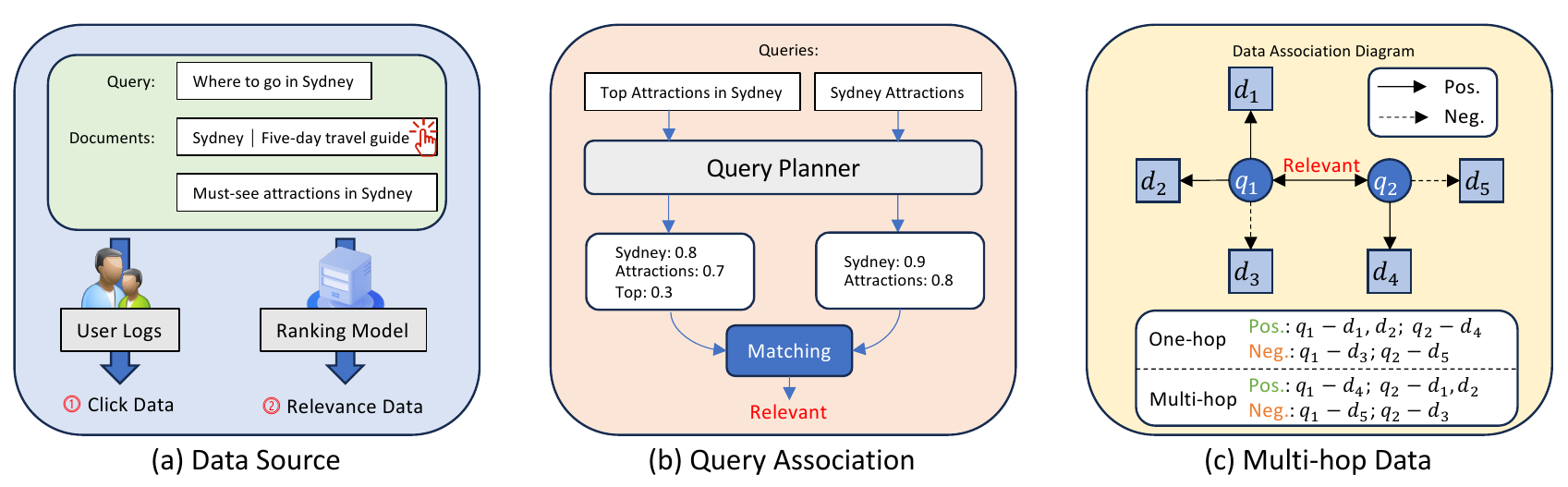}
    \caption{Key information about data collection includes: (a) two sources of query-document pairs: user click behaviors and ranking model evaluations. (b) query association construction method. (c) an example of a multi-hop data sample.}
    \label{fig:data_pipeline}
\end{figure*}

In online query retrieval, once the system receives a query from a user, the content of the query will be processed by the Query Planner, while the user ID will be used to retrieve the recent user's personalized information.
All processed data is sent to real-time servers to complete encoding.
The overall query features come from three sources.
The first source is the query information, including embeddings from our query tower and intent analysis of the query.
The second source is the user information, including the profile, historical behavior, etc.
The third source is the cross-features of query and user.
The query side also takes a neural network to merge all this data into the final embedding.
This final embedding searches for the most relevant documents using IVFPQ.
Retrieved documents are then transferred to the ranking stage.

%% file: chapters/05_exp.tex
\section{Experiments}

\subsection{Dataset}
\subsubsection{Training Data}
\ 
\label{data_construction}

\textbf{Data Collection.}
As shown in Figure~\ref{fig:data_pipeline} (a), in training data collection, we gather query-document pairs from two primary sources: click data and relevance data.
For click data, clicked documents are considered positive examples for the associated query. 
From the exposed but not clicked documents, we randomly select one as a hard negative example.
For relevance data, we utilize a high-accuracy cross-encoder ranking model~\cite{ma2024fine,liu2021pre} on our platform to measure the relevance between queries and documents.
For each query, the document with the highest relevance score is selected as the positive example.
To select the hard negatives, we first filter out the top-$K$ relevant documents to prevent false negative samples.
Then, we randomly choose a document ranked between top-$K$ and top-$T$ as the hard negative.

To increase data volume and improve relevance, we also consider multi-hop data relevance based on query association.
For query association, as shown in Figure~\ref{fig:data_pipeline} (b), we use a query planner from our platform to score term importance in queries.
Then, we utilize the important terms in queries to determine if there is a click association or relevance association, as shown in Algorithm~\ref{alg1}.
After computing query associations, we can construct the multi-hop data.
When $ClickA_{q_1\leftarrow q_2}$ is true, the click positive documents and click hard negative documents of query $q_2$ are also associated with query $q_1$.
Similarly, if $RelA$ is true, the relevance positive documents and relevance hard negative documents are associated bidirectionally between query $q_1$ and $q_2$.
Figure~\ref{fig:data_pipeline} (c) shows the difference between one-hop and multi-hop data.

\begin{algorithm}[t]
\caption{Query Association Matching}
\label{alg1}
\SetAlgoLined
\KwIn{$q_1, q_2, Planner, th_{click}, th_{rel}$}
\KwOut{$ClickA_{q_1\leftarrow q_2}, ClickA_{q_2\leftarrow q_1}, RelA$}

\textbf{Functions:}
\begin{itemize}[leftmargin=*]
    \item $Planner(q_1):$ \Return term-importance score dictionary
    \item $\text{Top2}(Scores)$: \Return top 2 scoring terms based on $Scores$
    \item $\text{Top4}(Scores)$: \Return top 4 scoring terms based on $Scores$
    \item $\text{AllAbove}(Terms, Scores, Threshold)$: \Return True if all $Term \in Terms$ have $Scores[Term]> Threshold$ else False
\end{itemize}

\textbf{Main Algorithm:}

\Indp
\nl $S_1 \leftarrow Planner(q_1)$\;
\nl $S_2 \leftarrow Planner(q_2)$\;

\nl \tcc{Click association for $q_1$ to $q_2$}
\nl $T_1^2 \leftarrow \text{Top2}(S_1)$\;
\nl $ClickA_{q_1\leftarrow q_2} \leftarrow \text{AllAbove}(T_1^2, S_2, th_{click})$\;

\nl \tcc{Click association for $q_2$ to $q_1$}
\nl $T_2^2 \leftarrow \text{Top2}(S_2)$\;
\nl $ClickA_{q_2\leftarrow q_1} \leftarrow \text{AllAbove}(T_2^2, S_1, th_{click})$\;

\nl \tcc{Relevance association}
\nl $T^4_1 \leftarrow \text{Top4}(S_1)$\;
\nl $T^4_2 \leftarrow \text{Top4}(S_2)$\;
\nl$RelA \leftarrow \text{AllAbove}(T^4_1, S_2, th_{rel}) \land \text{AllAbove}(T^4_2, S_1, th_{rel})$\;

\nl \Return $ClickA_{q_1\leftarrow q_2}, ClickA_{q_2\leftarrow q_1}, RelA$\;
\Indm
\end{algorithm}

Based on the above data collection, we can obtain datasets from four sources: click one-hop, click multi-hop, relevance one-hop and relevance multi-hop.
We merge all datasets as $\mathbb{T} = \{\langle q_i,d_i^+,d_i^-\rangle\}_{i=1}^{|\mathbb{T}|}$.

\textbf{Training Datasets.}
We employ a three-tiered approach using \textbf{Small}, \textbf{Median}, and \textbf{Large} datasets. 
The training and validation datasets are derived from the original datasets in a $9:1$ ratio.
We show the details of three datasets in Table~\ref{tab:train_datasets}.
Besides, all experiments on QKD are conducted with Large datasets because the training cost for the smaller student query towers is low.

\begin{table}[!t]
    \centering
    
    \setlength\tabcolsep{2mm}
    \renewcommand\arraystretch{1.05}
    \caption{Detailed statistics of training dataset.}
    \scalebox{0.8}{
    \begin{tabular}{l|r|l|r} \Xhline{1.0pt}
    \multicolumn{4}{c}{\textbf{Small training dataset}} \\
    \Xhline{1.0pt}
    \# documents & 1,642,559 & \# query-document pairs & 899,693 \\

    avg. \# words per title & 13.86 & avg. \# topic words per doc. & 25.18 \\

    avg. \# words per content & 200.33 & avg. \# words per query & 16.05  \\
    
    \Xhline{1.0pt}
    \multicolumn{4}{c}{\textbf{Median training dataset}} \\
    \Xhline{1.0pt}
    \# documents & 26,263,129 & \# query-document pairs & 19,638,050 \\

    avg. \# words per title & 13.56 & avg. \# topic words per doc. & 25.86\\

    avg. \# words per content & 180.21 & avg. \# words per query & 15.56 \\
    \Xhline{1.0pt}
    \multicolumn{4}{c}{\textbf{Large training dataset}} \\
    \Xhline{1.0pt}
    \# documents & 88,458,365 & \# query-document pairs & 117,828,000 \\

    avg. \# words per title & 13.58 & avg. \# topic words per doc. & 25.69\\

    avg. \# words per content & 182.26 & avg. \# words per query & 15.57 \\
    
    \Xhline{1.0pt}
    \end{tabular}}
    \label{tab:train_datasets}
\end{table}

\subsubsection{Testing Data}
\ 

\textbf{Data Collection.}
To evaluate the models, we create a test dataset following the same data collection process, but with data collected one month after the training dataset.
We add additional non-paired documents to the test dataset to simulate a realistic retrieval pool.
We call this test dataset the \textbf{Recall} test dataset.

While the Recall test dataset is useful for identifying documents potentially related to queries, it is often noisy and may not always accurately reflect true relevance.
Therefore, we create a manually curated dataset that can serve as a more reliable benchmark for assessing query-document relevance.
We collect query-document pairs following the same data collection process used for the Recall test dataset. Expert reviewers then evaluate these pairs based on two key criteria: satisfaction and relevance.
For each pair, experts assign a binary score: $0$ for unsatisfactory or irrelevant, and $1$ for satisfactory or relevant.
We call this dataset the \textbf{Manual} test dataset.

\textbf{Testing Datasets.}
We collect queries that do not appear in the training dataset, and show the details of two types of testing datasets in Table~\ref{tab:test_datasets}.

\textbf{Evaluation Settings.}
In the Recall dataset, we utilize the recall metric $R@K$ to evaluate the effectiveness of models.
Specifically, for each query, we rank all documents in the document pool based on their relevance scores.
We consider a retrieval successful at rank $K$ if the query-relevant ground truth document appears within the top-$K$. 
We denote the number of successful retrievals at rank $K$ across all queries as $Success@K$.
After processing all queries, we calculate $R@K$ as $R@K = \frac{Success@K}{\# Queries}$.
We report $R@50$, $R@100$, $R@500$, and $R@1000$ on the Recall dataset. 
In the Manual dataset, we compute similarity scores for all query-document pairs in the dataset, and rank all pairs based on these similarity scores.
We calculate AUC-SAT and AUC-REL using human-annotated satisfaction and relevance labels, respectively, along with computed rankings.

\begin{table}[!t]
    \centering
    
    \setlength\tabcolsep{2mm}
    \renewcommand\arraystretch{1.05}
    \caption{Detailed statistics of testing dataset.}
    \scalebox{0.8}{
    \begin{tabular}{l|r|l|r} \Xhline{1.0pt}
    \multicolumn{4}{c}{\textbf{Recall testing dataset}} \\
    \Xhline{1.0pt}
    \# documents & 1,007,166 & \# query-document pairs & 15,370 \\

    avg. \# words per title & 11.80 & avg. \# topic words per doc. & 21.43\\

    avg. \# words per content & 157.00 & avg. \# words per query & 15.92  \\
    
    \Xhline{1.0pt}
    \multicolumn{4}{c}{\textbf{Manual testing dataset}} \\
    \Xhline{1.0pt}
    \# documents & 14,448 & \# query-document pairs & 14,458 \\

    avg. \# words per title & 14.05 & avg. \# topic words per doc. & 24.81\\

    avg. \# words per content & 242.90 & avg. \# words per query & 15.34 \\
    \# satisfactory q.-d. pairs &9873 & \# unsatisfactory q.-d. pairs & 4586 \\
    \# relevant q.-d. pairs &9389 & \# irrelevant q.-d. pairs & 5070 \\
    \Xhline{1.0pt}
    \end{tabular}}
    \label{tab:test_datasets}
\end{table}

\subsection{Implementation Details}
We use Qwen-2.5 $7B$ model~\cite{qwen2.5} to scale up the dense retrieval models.
In our experiments, we use cosine similarity as the $sim()$ function for better training stability.
We set $dim=128$, $dim_{low}=16$, the overall batch size (combined batch size across all GPUs) as $480$, all $w^m$ and $w^m_{hard}$ as $1$, $margin=0.2$, $\alpha=0.5$, and $\lambda=1$.
The original query tower is based on the BERT architecture with just four layers and integrates an additional vocabulary from Xiaohongshu.
We presented the average results of three tests for reliability.
For more training and model details, please check Appendix~\ref{training_details}.

\begin{table}[!t]
    \centering
    
    \setlength\tabcolsep{4pt}
    \renewcommand\arraystretch{1}
    \caption{AUC performance of different methods trained on training datasets of different scales on the Manual test dataset (\%). ``\textbf{{\Large *}}'' indicates the statistically significant improvements (i.e., two-sided t-test with $p<0.05$) over \methoddoc.}
    \scalebox{0.85}{
    \begin{tabular}{c|l|ccc|c} 
    \Xhline{1.0pt}
    \textbf{Datasets}&\textbf{Metric}&TwinBERT&\methoddoc&\method&\methoddual \\
    \Xhline{0.5pt}
    \multirow{2}*{\textbf{Small}}&AUC-SAT&70.67&74.76&\textbf{77.58}*&78.89 \\
    ~&AUC-REL&69.24&73.16&\textbf{76.29}*&77.14 \\
    \Xhline{0.5pt}
    \multirow{2}*{\textbf{Median}}&AUC-SAT&76.54&80.00&\textbf{80.96}*&81.14 \\
    ~&AUC-REL&74.75&79.26&\textbf{79.61}*&80.28 \\
    \Xhline{0.5pt}
    \multirow{2}*{\textbf{Large}}&AUC-SAT&78.91&81.20&\textbf{83.01}*&83.71 \\
    ~&AUC-REL&77.76&80.48&\textbf{82.14}*&82.91 \\
    \Xhline{1.0pt}
    \end{tabular}}
    \label{tab:overall_auc}
\end{table}

\subsection{Overall Performance}
We compare the following methods to validate the efficacy of our \method.
\begin{itemize}[leftmargin=*]
    \item \textbf{TwinBERT}~\cite{lu2020twinbert} is a dense retrieval method that adopts the dual BERT~\cite{kenton2019bert} architecture, which is a traditional online method.
    \item \textbf{\methoddoc} is a scaling strategy that only scales the document tower while maintaining the query tower.
    \item \textbf{\methoddual} is a scaling strategy that fully scales the query and document towers with LLMs. This is an upper bound method, which is hard to deploy online.
\end{itemize}

This result is shown in Table~\ref{tab:overall_auc} and Table~\ref{tab:overall_recall}.
Firstly, we observe that all scaling methods show significant improvement compared to TwinBERT, as scaling with LLMs both enhances the retriever's knowledge and provides more parameters, enabling better generalization across diverse queries and documents.
Secondly, \methoddual shows superior performance compared to \methoddoc.
This is because the limited size of the query tower in \methoddoc restricts the full scalability of the LLM-based document tower.
Besides, \method outperforms \methoddoc, while \methoddoc has the same number of parameters as \method.
And \method has only a small performance gap with \methoddual.
This shows that \method not only fully realizes the scaling potential of the LLM-based document tower but also maintains the size of the query tower that requires online inference.

\begin{table*}[!h]
    \centering
    
    \setlength\tabcolsep{3pt}
    \renewcommand\arraystretch{1}
    \caption{Recall performance of different methods trained on training datasets of different scales on the Recall test dataset (\%). 
    ``\textbf{{\Large *}}'' indicates the statistically significant improvements (i.e., two-sided t-test with $p<0.05$) over \methoddoc.}
    \scalebox{0.8}{
    \begin{tabular}{l|ll|cccc|cccc|cccc} 
    \Xhline{1.0pt}
    \multirow{2}*{\textbf{Method}}&\multirow{2}*{\textbf{Query Tower}}&\multirow{2}*{\textbf{Doc. Tower}}&\multicolumn{4}{c|}{\textbf{Small}}&\multicolumn{4}{c|}{\textbf{Median}}&\multicolumn{4}{c}{\textbf{Large}} \\
    \cline{4-15}
    ~&~&~&R@50&R@100&R@500&R@1k&R@50&R@100&R@500&R@1k&R@50&R@100&R@500&R@1k \\
    \Xhline{0.5pt}
    TwinBERT&RED BERT-4L (29M)~\cite{kenton2019bert}&zh BERT-12L (86M)~\cite{kenton2019bert}& 6.54&10.14&23.19&31.14& 21.25 & 28.65 & 48.79 & 58.42 & 38.64 & 45.23 & 60.81 & 71.25 \\
    \methoddoc&RED BERT-4L (29M)~\cite{kenton2019bert}&Qwen2.5-Ins. (7B)~\cite{qwen2.5}& 27.33 & 36.82 & 59.11 & 67.59 & 41.86 & 52.89 & 74.52 & 81.68 & 47.19 & 58.54 & 78.36 & 84.37\\
    \method&RED BERT-4L (29M)~\cite{kenton2019bert}&Qwen2.5-Ins. (7B)~\cite{qwen2.5}& \textbf{33.62}* & \textbf{43.81}*& \textbf{65.85}*& \textbf{74.67}*& \textbf{47.90}*& \textbf{59.14}*& \textbf{79.93}*& \textbf{86.01}*& \textbf{54.57}* & \textbf{66.07}*& \textbf{85.00}* & \textbf{90.02}* \\
    \Xhline{0.5pt}
    \methoddual&Qwen2.5-Ins. (7B)~\cite{qwen2.5}&Qwen2.5-Ins. (7B)~\cite{qwen2.5}& 35.66 & 45.21 & 66.75 & 75.48 & 49.48 & 60.37 & 81.00 & 86.79 & 55.15 & 66.92 & 85.63 & 90.62 \\
    \Xhline{1.0pt}
    \end{tabular}}
    \label{tab:overall_recall}
\end{table*}

\begin{figure*}[!t]
    \centering
    \subfigure[Scaling Law for Model Size]{
    \includegraphics[width=0.32\textwidth]{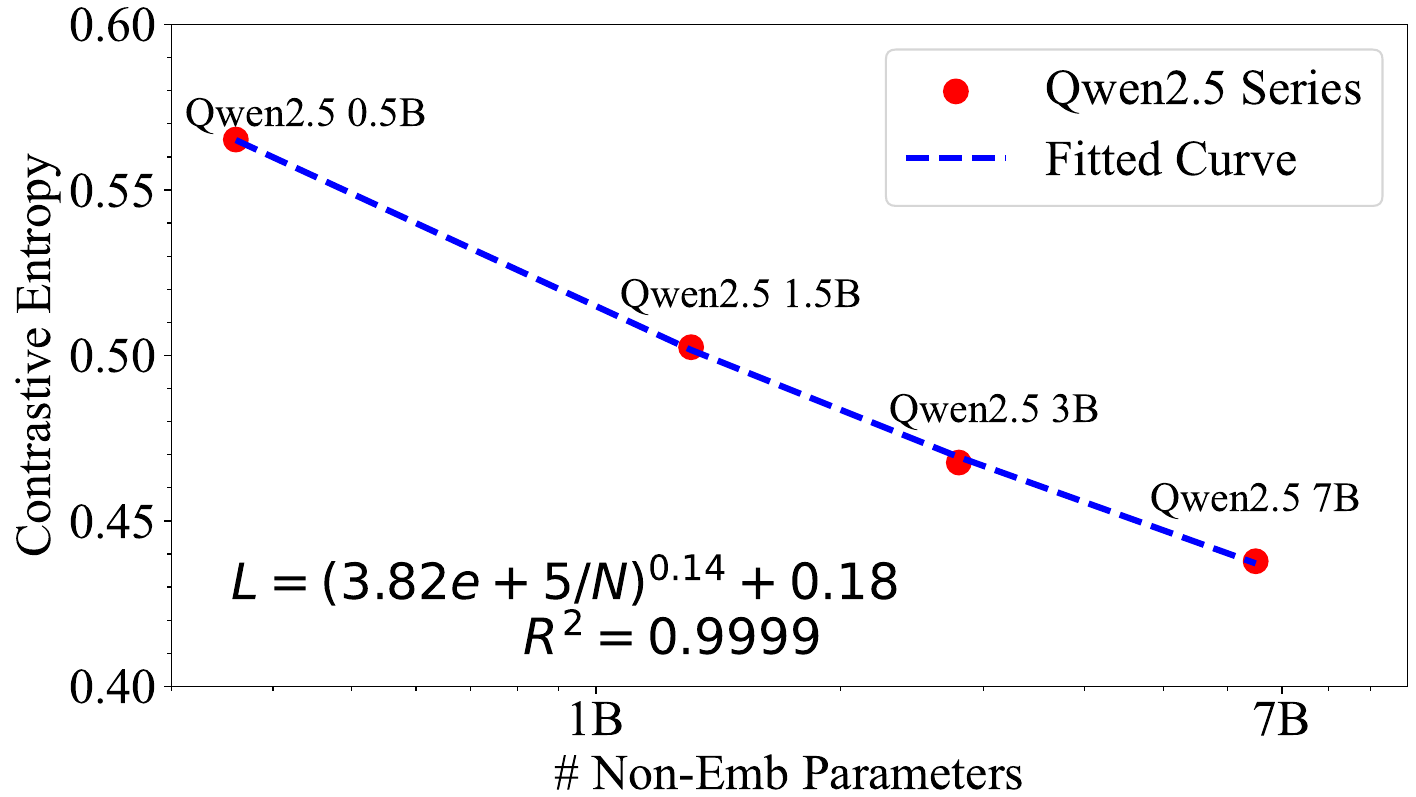}}
    \subfigure[Scaling Law for Data Size]{
    \includegraphics[width=0.32\textwidth]{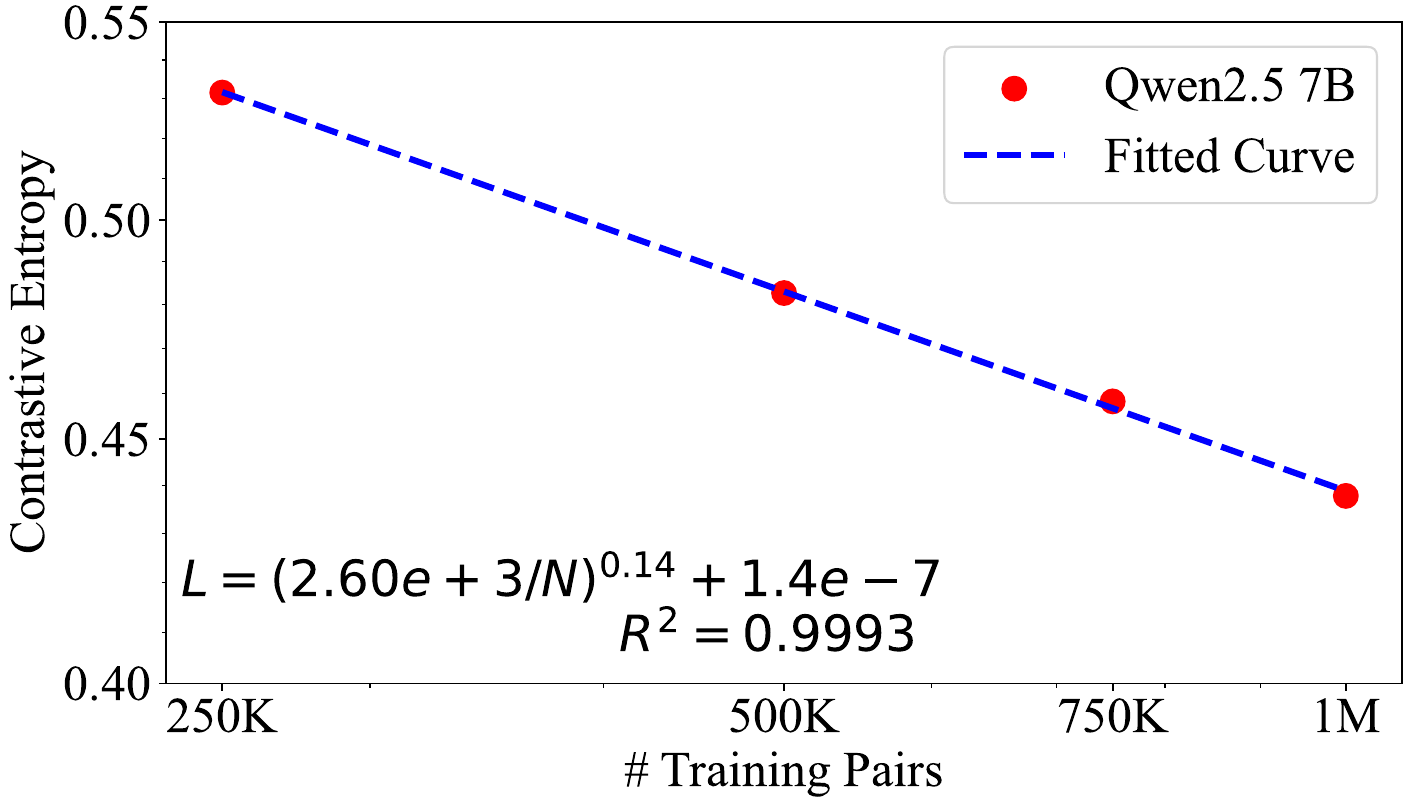}}
    \subfigure[Mixed Scaling Law of Model Size and Data Size.]{
    \includegraphics[width=0.315\textwidth]{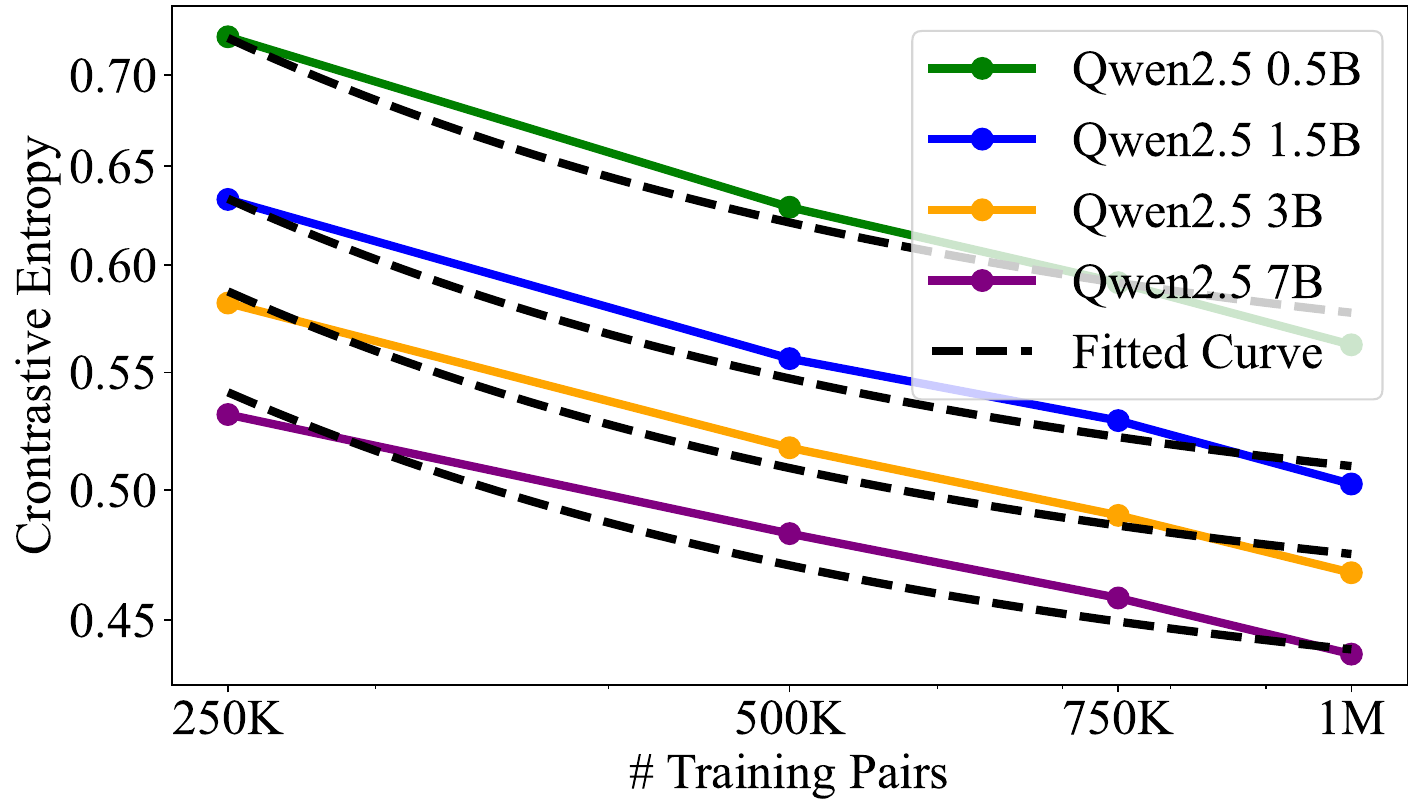}}
    \caption{The scaling laws of the LLM-based dual-tower architecture for real-world dense retrieval on Xiaohongshu. The dots are the practical experimental results. The dashed lines are the fitted curves of the scaling law. The y-axis is the contrastive entropy on the Small validation dataset. (a) The scaling law for model size. (b) The scaling law for data size. (c) The mixed scaling law of model size and data size.}
    \label{fig:scaling}
\end{figure*}

\subsection{Exploring Scaling Law}
Since the performance gap between the second stage and the first stage of \method is relatively small, we focus on exploring the scaling law of dense retrieval in the first stage on the Small training dataset.
We utilize advanced Qwen 2.5 series LLMs~\cite{qwen2.5} as the backbone models.
At the same time, we use real-world industrial scenarios and training pipelines to validate the scaling law.
We evaluate the retriever's ranking ability using contrastive entropy on the validation dataset~\cite{fang2024scaling} as our evaluation metric.
To analyze the scaling behavior across different factors, we fit our experimental data using the scaling law function~\cite{fang2024scaling,kaplan2020scaling,hoffmann2022training}:
$L(x) = (\frac{C}{x})^{\alpha} + \delta$,
where $x$ represents the scaling variable, and $L(x)$ is the contrastive entropy on the validation dataset to be predicted.
$\delta$, $C$ and $\alpha$ are coefficients. 
Besides, there is a coefficient of determination $R^2$ to judge the fit situation.
\subsubsection{Model Size Scaling Law}
The scaling law of model size in our platform is as follows:
\begin{equation}
L(N) = (\frac{3.82\times 10^{5}}{N})^{0.14} + 0.18,
\end{equation}
where $N$ is the number of non-embedding parameters in the retriever.
As shown in Figure~\ref{fig:scaling} (a), we find that the loss decreases as the model size scales in a power-law.

\subsubsection{Data Scaling Law}
To explore the data scaling law, we fix the model size and utilize the Qwen 2.5 $7$B LLM as the backbone of our dense retriever.
As shown in Figure~\ref{fig:scaling} (b), the scaling law of data size in our platform follows:
\begin{equation}
L(D) = (\frac{2.60\times 10^{3}}{D})^{0.14} + 1.4\times 10^{-7},
\end{equation}
where $D$ is the number of training query-document pairs.

\subsubsection{Model Size and Data Scaling Law}
The mixed scaling law of model size and data size considers these two factors simultaneously.
As shown in Figure~\ref{fig:scaling} (c), we utilize the following function to fit the validation loss for different model sizes and data sizes:
\begin{equation} 
L(N,D) = ((\frac{3.15 \times 10^5}{N})^{\frac{0.12}{770}} + \frac{1.16 \times 10^2}{D})^{770} + 0.108,
\label{eq:mix}
\end{equation}
where $R^2$ of this function equals $0.9887$.

\subsubsection{Analysis.} By analyzing Eq.~\ref{eq:mix},  we observe that scaling up data size $D$ is more efficient for improving model performance due to the large exponent.
However, the benefits of data scaling diminish gradually as $D$ approaches $10^8$. 
While increasing the training data could further improve performance, the additional training costs would significantly outweigh the potential benefits.
When the training data size is sufficiently large, scaling up the model size becomes a more effective way to improve performance.
Based on this insight and considering the training and inference costs, we used $10^8$ training pairs to train our $7B$ LLM-based dense retriever.

\begin{table}[!t]
    \centering
    
    \setlength\tabcolsep{3pt}
    \renewcommand\arraystretch{1}
    \caption{Performance of different prompt methods trained on the Small training dataset (\%). ``\textbf{{\Large *}}'' indicates the statistically significant improvements (i.e., two-sided t-test with $p<0.05$) over the best baseline.}
    \scalebox{0.85}{
    \begin{tabular}{l|cccc|cc} 
    \Xhline{1.0pt}
    \textbf{Method}&R@50&R@100&R@500&R@1k&AUC-SAT&AUC-REL \\
    \Xhline{0.5pt}
    \methoddual & \textbf{35.66}* & \textbf{45.21}* & \textbf{66.75}* & \textbf{75.48}* &  \textbf{78.89}* & \textbf{77.14}* \\
    w/o title loss & 25.71 & 38.05 & 61.37 & 71.22 & 75.91 & 74.88\\
    w/o summarize & 33.12 & 43.11 & 65.67 & 74.22 & 77.16 & 76.13 \\
    \Xhline{1.0pt}
    \end{tabular}}
    \label{tab:ablation}
\end{table}
\subsection{Ablation Study}
\subsubsection{Prompt Design.}
To test the efficacy of our proposed prompt, we compare our prompt design with the following variants.
\begin{itemize}[leftmargin=*]
    \item \textbf{w/o title loss} is a variant that does not predict the query terms for titles in the prompt template.
    \item \textbf{w/o summarize} is a variant that does not use an inductive prompt to summarize. Directly use the predicted query term embedding of the content as the final embedding.
\end{itemize}

We show the results in Table~\ref{tab:ablation}.
We observe that the performance of w/o title loss is lower than \methoddual.
This is because w/o title loss can overlook the information in titles, which is key to driving users to click documents.
We also find that w/o summarize can reduce the performance.
The inductive prompt can boost the LLMs to conclude the important information in the document again to reflect previous information, which is beneficial to final representations.

\begin{table}[!t]
    \centering
    
    \setlength\tabcolsep{3.3pt}
    \renewcommand\arraystretch{1.}
    \caption{Performance of student query towers trained with QKD on Large training dataset (\%), all of which are paired with the same 7B document tower of \methoddual trained on Large training dataset. \# Para. is the number of non-embedded parameters of query towers. QPS is tested on an A100 GPU, with an inference batch size of $500$ and a GPU utilization ratio of $100\%$. Train Loss is the mean QKD loss of the final student models tested on the Large training dataset. Valid Loss is tested on the Large validation dataset.}
    \scalebox{0.8}{
    \begin{tabular}{l|ccc|c} 
    \Xhline{1.0pt}
    \textbf{Metric}&RED BERT-1L&RED BERT-4L&RED BERT-12L&\methoddual \\
    \Xhline{0.5pt}
    \# Para. & 8M&29M&86M&7B \\
    QPS & 52,205 &33,810&19,090&408\\
    \Xhline{0.5pt}
    R@50 & 49.56 & 54.57 & 54.91 & 55.15  \\
    R@100 & 61.55 & 66.07 & 66.60 & 66.92  \\
    R@500 & 81.86 & 85.00 & 85.41 & 85.63 \\
    R@1k & 87.74 & 90.02 & 90.25 & 90.62 \\
    \Xhline{0.5pt}
    AUC-SAT&82.03&83.01&83.17&83.71 \\
    AUC-REL&81.22&82.14&82.53&82.91 \\
    \Xhline{0.5pt}
    Train Loss & 2.6e-2& 1.1e-4& 6.0e-3&- \\
    Valid Loss & 4.0e-2& 1.6e-2& 1.0e-2&- \\
    \Xhline{1.0pt}
    \end{tabular}}
    \label{tab:qkd}
\end{table}

\subsubsection{QKD for Varying Sizes of Student Query Tower}
We conduct QKD for different query towers of varying sizes, all of which are paired with the same 7B document tower of \methoddual trained on Large training dataset.
We show the results in Table~\ref{tab:qkd}.
We observe that the performance gap of the student query towers compared to \methoddual decreases as the number of parameters in the query towers increases, though the QPS decreases.
There is a trade-off between improved performance and online inference speed.
Besides, we find that Train Loss and Valid Loss are quite low when distilling any query towers, which is crucial to the theoretical upper bound of the expected risk of \method.

\subsection{Online Experiment}
We conduct week-long online A/B testing experiments on the retrieval system of Xiaohongshu, where the experimental traffic accounts for $5\%$ of the overall traffic.
Compared to the old online system that uses traditional BERT as the semantic interpreter, our new online method incorporating \method significantly enhances the relevance between retrieved documents and queries.
In the top 4 retrieved documents, the number of documents labeled as irrelevant by the ranking model decreases by \(1.165\%\), and the satisfaction level improves by $0.135\%$. In the top 20 documents, the number of documents labeled as irrelevant decreases by \(1.546\%\), and the satisfaction level improves by $0.172\%$.
Besides, the ratio of queries with fewer than 5 documents recalled decreases by $1.037\%$.
All these improvements are statistically significant (t-test with $p<0.05$).

%% file: chapters/02_related.tex
\section{Related Work}
\textbf{Dense Textual Retrieval.}
Dense textual retrieval employs query and document towers to transform queries and documents into embeddings~\cite{zhao2024dense,lu2020twinbert,li2021embedding}.
This approach allows for the decoupling of query and document processing, enabling offline computation of document embeddings and indexing~\cite{liu2021que2search,he2023que2engage,magnani2022semantic}.
Based on pre-trained language models, dense retrieval can comprehend the relevance between queries and documents, making it a classic retrieval method in industrial applications~\cite{zhao2024dense,zhang2022uni,liu2021pre,liu2021que2search,he2023que2engage,magnani2022semantic}. 
Recent research has demonstrated the scaling law of dense retrieval in academic experimental settings~\cite{fang2024scaling,ni2022large}. 
Furthermore, LLMs have been successfully integrated into dense retrieval systems~\cite{ma2024fine,liao2024d2llm,zhu2023large,behnamghader2024llm2vec,li2024llama2vec,tang2024pooling}.
These advancements motivate our exploration of scaling up real-world retrievers in industry settings to enhance the performance of dense retrieval.

\textbf{LLM-enhanced Retrieval.}
LLMs have significantly advanced the field of information retrieval, leading to various changes in its aspects.
For example, LLMs can rewrite user queries to better capture the user's true intent and expand the semantic content of queries~\cite{wang2023query2doc,peng2024large}.
Retrieval-augmented generation utilizes the generative capabilities of LLMs to synthesize query-relevant details from retrieved documents, providing more comprehensive responses to users~\cite{gao2023retrieval,lyu2024crud,lyu2024retrieve}.
Some works leverage LLMs to generate high-quality pseudo training data, enhancing model performance~\cite{wang2023improving}.
Additionally, numerous academic studies have leveraged LLMs as backbones to enhance the understanding of textual content~\cite{liao2024d2llm,zhu2023large,behnamghader2024llm2vec,li2024llama2vec,lee2024nv,ma2024task}.
In our work, we focus on applying LLMs to dense retrieval for real-world applications, taking into account online efficiency and memory usage constraints.

\textbf{Knowledge Distillation for Dense Retrieval.}
Retrieval systems require low online latency to satisfy users' needs. 
To achieve this, knowledge distillation has become a popular technique to compress large or powerful models into smaller and faster ones~\cite{zhang2022uni,hinton2015distillingknowledgeneuralnetwork}.
Cross-encoders, known for their high precision in capturing query-document interaction~\cite{ren2021rocketqav2,qu2021rocketqa}, are often used as teacher models of dual-tower~\cite{ren2021rocketqav2,qu2021rocketqa,asai2023task,kimustad,izacard2021distilling}.
This approach works well in both supervised~\cite{ren2021rocketqav2,liao2024d2llm} and unsupervised scenarios~\cite{qu2021rocketqa,asai2023task}.
Some works~\cite{lin2021batch,hofstatter2021efficiently} have also adopted late-interaction dual-tower models~\cite{khattab2020colbert} to distill naive dual-tower models.
Besides, other works have explored faster online inference by distilling the query tower for better online architecture~\cite{cohen-etal-2024-extremely} or faster inference~\cite{wang2023query,campos2023quick}.
In this work, we utilize the knowledge of the ranking model to generate relevance data for model training. 
To improve performance while maintaining efficiency, we propose a two-stage scaling approach.

%% file: chapters/06_conclusion.tex
\section{Conclusion}
In this work, we explore scaling up the dense retriever system for practical industrial applications.
To exploit the scaling potential of LLMs while maintaining online query latency, we propose a two-stage method, \method.
The first stage is fully scaling up the dual towers, which uncovers the full potential of LLMs for dense retrieval.
The second stage is to conduct QKD to reduce the online latency.
We conduct comprehensive experiments to validate the scaling law of dense retrieval in our specific scenario and demonstrate the efficacy of our proposed framework.
Besides, our model demonstrates significant online performance improvements.

%% file: chapters/07_appendix.tex
\section*{Appendix}

\section{Training Details}
\label{training_details}
\begin{table}[!h]
    \centering
    \setlength\tabcolsep{2mm}
    \renewcommand\arraystretch{0.9}
    \caption{Training hyperparameters for training \method.}
\begin{tabular}{lcc}
\Xhline{1.0pt}
\textbf{Configuration} & \textbf{Stage I} & \textbf{Stage II} \\
\Xhline{0.5pt}

Optimizer & AdamW& AdamW \\
$\beta_1$ & $0.9$ & $0.9$  \\
$\beta_2$ & $0.999$ & $0.999$ \\
$epsilon$ & $1e^{-8}$ & $1e^{-8}$ \\
Max gradient norm & $1.0$ & $1.0$ \\
Weight decay & $1e^{-3}$& $1e^{-3}$ \\
Peak learning rate & $3e^{-6}$ & $3e^{-4}$ \\
Warmup Ratio & $0.1$ & $0.1$\\
Learning rate scheduler & linear decay  & linear decay\\
Numerical precision & bf16& bf16  \\
Global batch size & $480$ & $40,000$ \\
Epoch & 1 & 100 \\

\Xhline{1.0pt}
\end{tabular}
\label{tab:traininghyper}
\end{table}

The detailed training hyperparameters are presented in Table~\ref{tab:traininghyper}.
We use DeepSpeed~\cite{rasley2020deepspeed}, and Zero Redundancy Optimizer (ZeRO)~\cite{ren2021zero} Stage 3 to train our models.
Besides, we provide detailed information on the training models used in our main experiments in Table~\ref{tab:modelhyper}.
More details about the Qwen2.5 series are presented in~\cite{qwen2.5}.

\begin{table}[!h]
    \centering
    \setlength\tabcolsep{2mm}
    \renewcommand\arraystretch{1}
    \caption{Hyperparameters of used models.}
\scalebox{0.85}{
\begin{tabular}{lccc}
\Xhline{1.0pt}
\textbf{Setting} & \textbf{RED Bert-4L} & \textbf{zh Bert-12L} & \textbf{Qwen2.5-Ins. 7B} \\
\Xhline{0.5pt}
\# layers & $4$ & $12$ & $28$ \\
\# attention heads & $12$ & $12$ & $28$ \\
Vocab size & $84,522$ & $21,128$ & $152,064$ \\
hidden size $h_t$ & $768$ & $768$ & $3,584$ \\
intermediate size & $3,072$ & $3,072$ & $18,944$ \\
\Xhline{1.0pt}
\end{tabular}}
\label{tab:modelhyper}
\end{table}

\section{Theoretical Proof}
\label{theory}
Following~\cite{kimustad}, we consider only one contrastive learning loss, based on one embedding generated by encoders, as the overall loss for simplicity.
In our \method, we adopt the softmax-based contrastive learning loss.
We reformulate the functions with relevance label $y_i$ as follows:
\begin{equation}
\ell(\boldsymbol{s}_{q,\mathbb{D}},\boldsymbol{y}) = -\sum_{i=1}^{|\mathbb{D}|} y_i \cdot log(\frac{e^{s_{q,d_i}}}{\sum_{j=1}^{|\mathbb{D}|}e^{s_{q,d_j}}}),
\end{equation}
where $\mathbb{D} = [d_1, ..., d_{|\mathbb{D}|}]$, $\boldsymbol{s}_{q,\mathbb{D}} = [s_{q,d_1}, ..., s_{q,d_{|\mathbb{D}|}}]$, and $\boldsymbol{y}=[y_1,...,y_{|\mathbb{D}|}]$.
$\mathbb{D}$ is the set of candidate documents.
$\boldsymbol{s}_{q,\mathbb{D}}$ is the set of relevance scores for the candidate documents, and $s_{q,d_i}$ is the relevance score between query $q$ and document $d_i$.
$\boldsymbol{y}$ is the set of relevance labels, where $y_i\in\{0,1\}$ indicates the relevance label between query $q$ and document $d_i$. 
It has a common alternative that changes from list-wise loss to point-wise loss~\cite{kimustad}, as follows:
\begin{equation}
\begin{aligned}
\ell(\boldsymbol{s}_{q,\mathbb{D}},\boldsymbol{y})&=\sum_{i=1}^{|\mathbb{D}|}\ell(s_{q,d_i},y_i) \\
& = -\sum_{i=1}^{|\mathbb{D}|}( y_i \cdot log(\frac{1}{1 + e^{-s_{q,d_i}}}) + (1-y_i) \cdot log(\frac{1}{1 + e^{s_{q,d_i}}})).
\end{aligned}
\end{equation}
This point-wise loss can also be rewritten as:
\begin{equation}
\begin{aligned}
\ell(s_{q,d_i},y_i) = (1-y_i)s_{q,d_i} + \gamma(-s_{q,d_i}),
\end{aligned}
\end{equation}
where $\gamma(s) = log (1+e^s)$ is the softplus function.
Subsequent analysis is based on the point-wise loss function to simplify the analysis.
And we simplify the relevance score as dot product $s_{q,d_i} = \boldsymbol{q}^\top \boldsymbol{d}_i$, because cosine similarity is essentially a scaled version of the dot product.
These simplifications maintain the core of our analysis while reducing computational complexity.
Next, we introduce the theorem of uniform deviation bound in ~\cite{kimustad}.

\begin{theorem}
\label{theorem1}
Denote $\mathcal{Q}$ as the class functions of the query tower $Q$, and $\mathcal{D}$ as the class functions of the document tower $D$.
Given the training dataset $\mathbb{S}_n=\{(q_i,d_i,y_i)\}_{i=1}^{n}$, the $\epsilon$-covering number of a function class $\mathcal{H}$ under $L_2(\mathbb{P}_n)$ norm is denoted as $N(\epsilon, \mathcal{H})$, where $\| h \|_{L_2(\mathbb{P}_n)}^2 := \| h \|_n^2 := \frac{1}{n} \sum_{i=1}^{n} \|h(q_i, d_i)\|^2_2$.
Let $\ell$ be the loss function, which is $L_{\ell}$-Lipschitz in its non-target variable.
Assume that the embedding functions in $\mathcal{Q}$ and $\mathcal{D}$ output embeddings with $L_2$ norms at most $K$.
Define the uniform deviation as follows:

\begin{equation}
\begin{aligned}
\mathcal{E}_n(\mathcal{Q},\mathcal{D}) &= \sup_{Q\in\mathcal{Q},D\in\mathcal{D}}\left|\frac{1}{n}\sum_{i\in[n]}\ell(s^{Q,D}_{q_i,d_i},y_i) - \mathbb{E}_{q,d}\ell(s^{Q,D}_{q,d},y)\right| \\
& = \sup_{Q\in\mathcal{Q},D\in\mathcal{D}}\left|R(s_{q,d}^{Q,D};\mathbb{S}_n) - \mathbb{E}_{q,d}\ell(s^{Q,D}_{q,d},y)\right|,
\end{aligned}
\end{equation}
where $R(s_{q,d}^{Q,D};\mathbb{S}_n) = \frac{1}{n}\sum_{i\in[n]}\ell(s^{Q,D}_{q_i,d_i},y_i)$.
Then, we have:
\begin{equation}
\begin{aligned}
    \mathcal{E}_n(\mathcal{Q},\mathcal{D})\leq \mathbb{E}_{\mathbb{S}_n} \frac{48KL_{\ell}}{\sqrt{n}} \int_0^{\infty} \sqrt{\log N(u, \mathcal{Q}) + \log N(u, \mathcal{D})}\ du.
\end{aligned}
\label{eall}
\end{equation}
For a fixed document tower $D^*$, i.e., $\mathcal{D}=\{D^*\}$,
\begin{equation}
\begin{aligned}
    \mathcal{E}_n(\mathcal{Q},\{D^*\})\leq \mathbb{E}_{\mathbb{S}_n} \frac{48KL_{\ell}}{\sqrt{n}} \int_0^{\infty} \sqrt{\log N(u, \mathcal{Q})}\ du.
\end{aligned}
\label{edis}
\end{equation}
\end{theorem}

Based on Theorem~\ref{theorem1}, we can prove Proposition~\ref{pro1} and Proposition~\ref{pro2}.
\begin{proof}[Proof of Proposition~\ref{pro1}]
Rewrite Eq.~\ref{eall} as follows:
\begin{equation}
\begin{aligned}
    &\sup_{Q_D\in\mathcal{Q},D_D\in\mathcal{D}}\left|R(s_{q,d}^{Q_D,D_D};\mathbb{S}_n) - \mathbb{E}_{q,d}\ell(s^{Q_D,D_D}_{q,d},y)\right|\leq \\& \mathbb{E}_{\mathbb{S}_n} \frac{48KL_{\ell}}{\sqrt{n}} \int_0^{\infty} \sqrt{\log N(u, \mathcal{Q}) + \log N(u, \mathcal{D})}\ du.
\end{aligned}
\label{eall_d}
\end{equation}
To derive the upper bound of the expected risk, we break down Eq.~\ref{eall_d} as follows:
\begin{equation}
\begin{aligned}
    &\mathbb{E}_{q,d}\ell(s^{Q_D,D_D}_{q,d},y)-R(s_{q,d}^{Q_D,D_D};\mathbb{S}_n)\leq \\& \mathbb{E}_{\mathbb{S}_n} \frac{48KL_{\ell}}{\sqrt{n}} \int_0^{\infty} \sqrt{\log N(u, \mathcal{Q}) + \log N(u, \mathcal{D})}\ du.
\end{aligned}
\end{equation}
After transferring $R(s_{q,d}^{Q_D,D_D};\mathbb{S}_n)$, we complete the proof.
\end{proof}

\begin{proof}[Proof of Proposition~\ref{pro2}]
Due to the student query tower $Q_{stu}$ sharing the same document tower $D_R$ with the teacher query tower $Q_R$, the document tower $D_R$ is fixed.
Therefore, rewrite Eq.~\ref{edis} as follows:
\begin{equation}
\begin{aligned}
    &\sup_{Q_{stu}\in\mathcal{Q}}\left|R(s_{q,d}^{Q_{stu},D_R};\mathbb{S}_n) - \mathbb{E}_{q,d}\ell(s^{Q_{stu},D_R}_{q,d},y)\right|\leq \\& \mathbb{E}_{\mathbb{S}_n} \frac{48KL_{\ell}}{\sqrt{n}} \int_0^{\infty} \sqrt{\log N(u, \mathcal{Q})}\ du.
\end{aligned}
\end{equation}
We can derive the upper bound of the expected risk by following Proposition~\ref{pro1}:
\begin{equation}
\begin{aligned}
    \mathbb{E}_{q,d}\ell(s^{Q_{stu},D_R}_{q,d},y)&\leq R(s_{q,d}^{Q_{stu},D_R};\mathbb{S}_n)   \\
    &+ \mathbb{E}_{\mathbb{S}_n} \frac{48KL_{\ell}}{\sqrt{n}} \int_0^{\infty} \sqrt{\log N(u, \mathcal{Q})}\ du.
\end{aligned}
\end{equation}
However, the student query tower $Q_{stu}$ and the document tower $D_R$ are not directly trained on the training dataset $\mathbb{S}_n$.
It is hard to measure $R(s_{q,d}^{Q_{stu},D_R};\mathbb{S}_n)$.
To clearly derive the upper bound, we introduce $R(s_{q,d}^{Q_R,D_R};\mathbb{S}_n)$, which is the empirical risk of jointly optimized teacher query tower and document tower:
\begin{equation}
\begin{aligned}
    \mathbb{E}_{q,d}\ell(s^{Q_{stu},D_R}_{q,d},y) &\leq R(s_{q,d}^{Q_R,D_R};\mathbb{S}_n) + \underbrace{R(s_{q,d}^{Q_{stu},D_R};\mathbb{S}_n) - R(s_{q,d}^{Q_R,D_R};\mathbb{S}_n) }_{\textit{Denoted as }L_{distill}} \\ &  + \mathbb{E}_{\mathbb{S}_n} \frac{48KL_{\ell}}{\sqrt{n}} \int_0^{\infty} \sqrt{\log N(u, \mathcal{Q})}\ du.
\end{aligned}
\label{infer_distill}
\end{equation}
Here, $L_{distill}$ represents the gap in empirical risks between the student query tower and the teacher query tower with respect to the same document tower.
\begin{equation}
\begin{aligned}
    L_{distill} &= R(s_{q,d}^{Q_{stu},D_R};\mathbb{S}_n) - R(s_{q,d}^{Q_R,D_R};\mathbb{S}_n) \\
    &= \frac{1}{n}\sum_{i\in[n]}\ell(s^{Q_{stu},D_R}_{q_i,d_i},y_i) - \frac{1}{n}\sum_{i\in[n]}\ell(s^{Q_R,D_R}_{q_i,d_i},y_i) \\
    &= \frac{1}{n}\sum_{i\in[n]} ((1-y_i)\boldsymbol{q}_i^{stu\top}\boldsymbol{d}_i + \gamma(-\boldsymbol{q}_i^{stu\top}\boldsymbol{d}_i) \\
    &- (1-y_i)\boldsymbol{q}_i^{\top}\boldsymbol{d}_i - \gamma(-\boldsymbol{q}_i^{\top}\boldsymbol{d}_i) ) \\
    &\overset{(a)}\leq \frac{1}{n}\sum_{i\in[n]} ((1-y_i)\boldsymbol{d}_i^\top(\boldsymbol{q}_i^{stu}-\boldsymbol{q}_i)+\left| \boldsymbol{q}_i^{stu\top}\boldsymbol{d}_i - \boldsymbol{q}_i^{\top}\boldsymbol{d}_i \right|) \\
    &\overset{(b)}\leq \frac{1}{n}\sum_{i\in[n]} ((1-y_i)\|\boldsymbol{d}_i\|\|\boldsymbol{q}_i^{stu}-\boldsymbol{q}_i\|+\| \boldsymbol{d}_i \|\| \boldsymbol{q}_i^{stu}-\boldsymbol{q}_i \| ) \\
    &\leq \frac{K}{n}\sum_{i\in[n]} ((2-y_i)\|\boldsymbol{q}_i^{stu}-\boldsymbol{q}_i\| ) \\
    &\leq \frac{2K}{n}\sum_{i\in[n]} \|\boldsymbol{q}_i^{stu}-\boldsymbol{q}_i\|, \\
\end{aligned}
\label{l_distill}
\end{equation}
where $(a)$ is because $\gamma$ is a Lipschitz continuous function, and $(b)$ is Cauchy-Schwarz inequality.
And, $\boldsymbol{q}_i^{stu}$ and $\boldsymbol{q}_i$ are embeddings generated by the student query tower $Q_{stu}$ and the teacher query tower $Q_R$ for query $q_i$, and $\boldsymbol{d}_i$ is the embedding produced by the document tower $D_R$ for document $d_i$.
Merging Eq.~\ref{infer_distill} and~\ref{l_distill} can finish the proof.
\end{proof}

\section{Case Study}
In this section, we present cases comparing the Online Baseline and \method, as shown in Table~\ref{tab:case}.
We observe three main distinctions.
Firstly, \method shows greater robustness to typos.
In the first case, it correctly identifies "Peony" despite the "Peany" typo, while the Online Baseline fails.
Secondly, scaling retrievers with LLMs improves accuracy in world knowledge. 
The second case shows \method correctly identifying "Mountain Climbing Mouse" as a brand, while the Online Baseline misinterprets it.
Lastly, \method better comprehends complicated queries. 
In the third case, it grasps the concept of how to send gifts to customers, while the Online Baseline focuses solely on "Zongzi".

\begin{table*}[!h]
    \centering
    
    \setlength\tabcolsep{9pt}
    \renewcommand\arraystretch{0.9}
    \caption{Case study comparing Online Baseline and \method. Blue text indicates clues related to the query, while red text indicates incorrect clues.}
    \scalebox{1}{
    \begin{tabular}{|p{3cm}|p{6cm}|p{6cm}|} 
    \Xhline{1.0pt}
    \textbf{Query (Explanation)}&\textbf{Online Baseline}&\textbf{\method} \\
    \Xhline{1.0pt}
    \multirow{2}{3cm}{\textcolor{blue}{Peany flower} painting. (There is a typo in this query. The word "Peany" should be spelled "Peony".)}&(\textbf{Rank 1st}) \textit{Title}: \textcolor{red}{Gesang flower}. \textit{Content}: Watercolor Flowers. \textit{Topic}: Painted Watercolor, \textcolor{red}{Watercolor Flowers}, \textcolor{red}{Gesang Flowers}. &(\textbf{Rank 1st}) \textit{Title}: \textcolor{blue}{Peony} Collection is here. \textit{Content}: Drawing Tool: iPad. Drawing Software: Procreate. Line Art and Coloring Brushes: Studio Brushes. \\
    \cline{2-3}
    ~&(\textbf{Rank 5th}) \textit{Title}: Draw two \textcolor{red}{roses}. \textit{Content}: ... I bought a box of rosa defective watercolor blind box. The colors are really practical ... Every time I paint the first flower, it kind of goes wrong. \textit{Topic}: \textcolor{red}{Watercolor Flowers}, Freestyle Watercolor, Watercolor Plants ... &(\textbf{Rank 5th}) \textit{Title}: Beginner Illustration Tutorial | New Chinese Style Illustration | How to \textcolor{blue}{Draw Peonies}. \textit{Content}: The \textcolor{blue}{peony} is known as the king of flowers, with unparalleled beauty ... Drawing Process: 1. Sketch: ... 2. Adjust Colors + Base Colors: ... 3. Shaping: ... Tools: iPad + Apple Pencil 2nd Gen + Procreate \textit{Topic}: Inspiration Notes, New Chinese Style, \textcolor{blue}{Draw Peony}.\\
    \Xhline{1.0pt}
    \multirow{2}{3cm}{Hong Kong \textcolor{blue}{Mountain Climbing Mouse}. ("Mountain Climbing Mouse" is the English translation of the Swedish outdoor clothing brand "Klättermusen".)}&(\textbf{Rank 1st}) \textit{Title}: Hong Kong Island's super niche parent-child adventure route - \textcolor{red}{mountain climbing}, waterfall fishing. \textit{Content}: This hidden hiking route is always on my list of places to take my kids out. First of all, it is very close to home ... In addition, although the stone road seems difficult to \textcolor{red}{climb}, it is actually very flat ... \textit{Topic}: Good Places to Take Kids Out in Hong Kong, Parent-Child Activities, \textcolor{red}{Mountain Climbing}, Travel. &(\textbf{Rank 1st})\textit{Title}: Back Nature on the 16th floor of Gala Place. \textit{Content}: If you like outdoor clothes, you must go to Back Nature on the 16th floor ... The bloggers' recommendations were indeed not misleading. They have \textcolor{blue}{Mountain Climbing Mouse}, Arc'teryx, Montbell, and Patagonia... \textit{Topic}: Hong Kong Travel, Hong Kong Shopping Sharing, Arc'teryx, \textcolor{blue}{Mountain Climbing Mouse}. \\
    \cline{2-3}
   ~&(\textbf{Rank 5th}) \textit{Title}: \textcolor{blue}{Mountain Climbing Mouse} Giant Island 30\% off. \textit{Content}: Website discount \textcolor{blue}{Mountain Climbing Mouse} 30\% off, women's sizes XXS, XS, S, Denmark direct mail, including tax. \textit{Topic}: \textcolor{blue}{klattermusen Mountain Climbing Mouse}, \textcolor{blue}{Mountain Climbing Mouse}. &(\textbf{Rank 5th}) \textit{Content}: Backpack. Gilling Backpack 26L 10292. Chinese name, Approximate weight of product: Backpack 26L (Frost Giant Gilling) 750g ... Very large U-shaped opening design, more convenient for taking and putting items. ... \textit{Topic}: Outdoor Sports, \textcolor{blue}{Mountain Climbing Mouse}, Backpack. \\
    \Xhline{1.0pt}
    \multirow{2}{3cm}{How to \textcolor{blue}{send Zongzi to customers}? ("Zongzi" is a traditional Chinese food. It is especially popular during the Dragon Boat Festival.)}&(\textbf{Rank 1st}) \textit{Title}: Happy \textcolor{red}{Dragon Boat Festival}. \textit{Topic}: \textcolor{red}{Happy Dragon Boat Festival}, \textcolor{red}{Dragon Boat Festival Zongzi}, Souvenirs. ... &(\textbf{Rank 1st})\textit{Title}: Workers, look here! To \textcolor{blue}{maintain customer relationships}, just bring it with you. \textit{Content}: \textcolor{blue}{Maintaining customer relationships} is something workers need to learn! How can you \textcolor{blue}{send big gifts with a small budget}. ... \textit{Topic}: \textcolor{blue}{Dragon Boat Festival gift}, Give gifts to customers, Gifts for leaders. ... \\
    \cline{2-3}
   ~&(\textbf{Rank 5th}) \textit{Content}: The best greetings are delivered first. \textcolor{red}{Dragon Boat Festival} is coming soon. Have you chosen your \textcolor{red}{Zongzi gift box}. \textit{Topic}: Dragon Boat Festival gift box, \textcolor{red}{Zongzi}... &(\textbf{Rank 5th}) \textit{Title}: 11 unspoken rules to remember when \textcolor{blue}{giving gifts during the Dragon Boat Festival}. \textit{Content}: 11 unspoken rules for giving gifts during the Dragon Boat Festival. The \textcolor{blue}{Dragon Boat Festival} is coming soon. Hurry up and prepare! \textit{Topic}: Dragon Boat Festival, \textcolor{blue}{Dragon Boat Festival gifts}, \textcolor{blue}{Giving gifts strategy}. \\
   \Xhline{1.0pt}
    \end{tabular}}
    \label{tab:case}
\end{table*}

%% file: sample-sigconf.bbl

\begin{thebibliography}{59}


\ifx \showCODEN    \undefined \def \showCODEN     #1{\unskip}     \fi
\ifx \showDOI      \undefined \def \showDOI       #1{#1}\fi
\ifx \showISBNx    \undefined \def \showISBNx     #1{\unskip}     \fi
\ifx \showISBNxiii \undefined \def \showISBNxiii  #1{\unskip}     \fi
\ifx \showISSN     \undefined \def \showISSN      #1{\unskip}     \fi
\ifx \showLCCN     \undefined \def \showLCCN      #1{\unskip}     \fi
\ifx \shownote     \undefined \def \shownote      #1{#1}          \fi
\ifx \showarticletitle \undefined \def \showarticletitle #1{#1}   \fi
\ifx \showURL      \undefined \def \showURL       {\relax}        \fi
\providecommand\bibfield[2]{#2}
\providecommand\bibinfo[2]{#2}
\providecommand\natexlab[1]{#1}
\providecommand\showeprint[2][]{arXiv:#2}

\bibitem[Asai et~al\mbox{.}(2023)]%
        {asai2023task}
\bibfield{author}{\bibinfo{person}{Akari Asai}, \bibinfo{person}{Timo Schick}, \bibinfo{person}{Patrick Lewis}, \bibinfo{person}{Xilun Chen}, \bibinfo{person}{Gautier Izacard}, \bibinfo{person}{Sebastian Riedel}, \bibinfo{person}{Hannaneh Hajishirzi}, {and} \bibinfo{person}{Wen-tau Yih}.} \bibinfo{year}{2023}\natexlab{}.
\newblock \showarticletitle{Task-aware Retrieval with Instructions}. In \bibinfo{booktitle}{\emph{Findings of ACL}}. \bibinfo{pages}{3650--3675}.
\newblock


\bibitem[BehnamGhader et~al\mbox{.}(2024)]%
        {behnamghader2024llm2vec}
\bibfield{author}{\bibinfo{person}{Parishad BehnamGhader}, \bibinfo{person}{Vaibhav Adlakha}, \bibinfo{person}{Marius Mosbach}, \bibinfo{person}{Dzmitry Bahdanau}, \bibinfo{person}{Nicolas Chapados}, {and} \bibinfo{person}{Siva Reddy}.} \bibinfo{year}{2024}\natexlab{}.
\newblock \showarticletitle{Llm2vec: Large language models are secretly powerful text encoders}.
\newblock \bibinfo{journal}{\emph{arXiv preprint arXiv:2404.05961}} (\bibinfo{year}{2024}).
\newblock


\bibitem[Brown et~al\mbox{.}(2020)]%
        {brown2020language}
\bibfield{author}{\bibinfo{person}{Tom~B Brown}, \bibinfo{person}{Benjamin Mann}, \bibinfo{person}{Nick Ryder}, \bibinfo{person}{Melanie Subbiah}, \bibinfo{person}{Jared Kaplan}, \bibinfo{person}{Prafulla Dhariwal}, \bibinfo{person}{Arvind Neelakantan}, \bibinfo{person}{Pranav Shyam}, \bibinfo{person}{Girish Sastry}, \bibinfo{person}{Amanda Askell}, {et~al\mbox{.}}} \bibinfo{year}{2020}\natexlab{}.
\newblock \showarticletitle{Language models are few-shot learners}. In \bibinfo{booktitle}{\emph{NeurIPS}}. \bibinfo{pages}{1877--1901}.
\newblock


\bibitem[Campos et~al\mbox{.}(2023)]%
        {campos2023quick}
\bibfield{author}{\bibinfo{person}{Daniel Campos}, \bibinfo{person}{Alessandro Magnani}, {and} \bibinfo{person}{ChengXiang Zhai}.} \bibinfo{year}{2023}\natexlab{}.
\newblock \showarticletitle{Quick dense retrievers consume kale: Post training kullback leibler alignment of embeddings for asymmetrical dual encoders}.
\newblock \bibinfo{journal}{\emph{arXiv preprint arXiv:2304.01016}} (\bibinfo{year}{2023}).
\newblock


\bibitem[Cohen and Kushilevitz(2024)]%
        {cohen-etal-2024-extremely}
\bibfield{author}{\bibinfo{person}{Fairstein~Yaron Cohen, Nachshon} {and} \bibinfo{person}{Guy Kushilevitz}.} \bibinfo{year}{2024}\natexlab{}.
\newblock \showarticletitle{Extremely efficient online query encoding for dense retrieval}. In \bibinfo{booktitle}{\emph{Findings of NAACL}}. \bibinfo{pages}{43--50}.
\newblock


\bibitem[Fang et~al\mbox{.}(2024)]%
        {fang2024scaling}
\bibfield{author}{\bibinfo{person}{Yan Fang}, \bibinfo{person}{Jingtao Zhan}, \bibinfo{person}{Qingyao Ai}, \bibinfo{person}{Jiaxin Mao}, \bibinfo{person}{Weihang Su}, \bibinfo{person}{Jia Chen}, {and} \bibinfo{person}{Yiqun Liu}.} \bibinfo{year}{2024}\natexlab{}.
\newblock \showarticletitle{Scaling laws for dense retrieval}. In \bibinfo{booktitle}{\emph{SIGIR}}. \bibinfo{pages}{1339--1349}.
\newblock


\bibitem[Gao et~al\mbox{.}(2023)]%
        {gao2023retrieval}
\bibfield{author}{\bibinfo{person}{Yunfan Gao}, \bibinfo{person}{Yun Xiong}, \bibinfo{person}{Xinyu Gao}, \bibinfo{person}{Kangxiang Jia}, \bibinfo{person}{Jinliu Pan}, \bibinfo{person}{Yuxi Bi}, \bibinfo{person}{Yi Dai}, \bibinfo{person}{Jiawei Sun}, {and} \bibinfo{person}{Haofen Wang}.} \bibinfo{year}{2023}\natexlab{}.
\newblock \showarticletitle{Retrieval-augmented generation for large language models: A survey}.
\newblock \bibinfo{journal}{\emph{arXiv preprint arXiv:2312.10997}} (\bibinfo{year}{2023}).
\newblock


\bibitem[He et~al\mbox{.}(2023)]%
        {he2023que2engage}
\bibfield{author}{\bibinfo{person}{Yunzhong He}, \bibinfo{person}{Yuxin Tian}, \bibinfo{person}{Mengjiao Wang}, \bibinfo{person}{Feier Chen}, \bibinfo{person}{Licheng Yu}, \bibinfo{person}{Maolong Tang}, \bibinfo{person}{Congcong Chen}, \bibinfo{person}{Ning Zhang}, \bibinfo{person}{Bin Kuang}, {and} \bibinfo{person}{Arul Prakash}.} \bibinfo{year}{2023}\natexlab{}.
\newblock \showarticletitle{Que2engage: Embedding-based retrieval for relevant and engaging products at facebook marketplace}. In \bibinfo{booktitle}{\emph{WWW}}. \bibinfo{pages}{386--390}.
\newblock


\bibitem[Hinton et~al\mbox{.}(2015)]%
        {hinton2015distillingknowledgeneuralnetwork}
\bibfield{author}{\bibinfo{person}{Geoffrey Hinton}, \bibinfo{person}{Oriol Vinyals}, {and} \bibinfo{person}{Jeff Dean}.} \bibinfo{year}{2015}\natexlab{}.
\newblock \bibinfo{title}{Distilling the Knowledge in a Neural Network}.
\newblock
\newblock
\showeprint[arxiv]{1503.02531}


\bibitem[Hoffmann et~al\mbox{.}(2022)]%
        {hoffmann2022training}
\bibfield{author}{\bibinfo{person}{Jordan Hoffmann}, \bibinfo{person}{Sebastian Borgeaud}, \bibinfo{person}{Arthur Mensch}, \bibinfo{person}{Elena Buchatskaya}, \bibinfo{person}{Trevor Cai}, \bibinfo{person}{Eliza Rutherford}, \bibinfo{person}{Diego de Las~Casas}, \bibinfo{person}{Lisa~Anne Hendricks}, \bibinfo{person}{Johannes Welbl}, \bibinfo{person}{Aidan Clark}, {et~al\mbox{.}}} \bibinfo{year}{2022}\natexlab{}.
\newblock \showarticletitle{Training compute-optimal large language models}. In \bibinfo{booktitle}{\emph{NeurIPS}}. \bibinfo{pages}{30016--30030}.
\newblock


\bibitem[Hofst{\"a}tter et~al\mbox{.}(2021)]%
        {hofstatter2021efficiently}
\bibfield{author}{\bibinfo{person}{Sebastian Hofst{\"a}tter}, \bibinfo{person}{Sheng-Chieh Lin}, \bibinfo{person}{Jheng-Hong Yang}, \bibinfo{person}{Jimmy Lin}, {and} \bibinfo{person}{Allan Hanbury}.} \bibinfo{year}{2021}\natexlab{}.
\newblock \showarticletitle{Efficiently teaching an effective dense retriever with balanced topic aware sampling}. In \bibinfo{booktitle}{\emph{SIGIR}}. \bibinfo{pages}{113--122}.
\newblock


\bibitem[Huang et~al\mbox{.}(2020)]%
        {huang2020embedding}
\bibfield{author}{\bibinfo{person}{Jui-Ting Huang}, \bibinfo{person}{Ashish Sharma}, \bibinfo{person}{Shuying Sun}, \bibinfo{person}{Li Xia}, \bibinfo{person}{David Zhang}, \bibinfo{person}{Philip Pronin}, \bibinfo{person}{Janani Padmanabhan}, \bibinfo{person}{Giuseppe Ottaviano}, {and} \bibinfo{person}{Linjun Yang}.} \bibinfo{year}{2020}\natexlab{}.
\newblock \showarticletitle{Embedding-based retrieval in facebook search}. In \bibinfo{booktitle}{\emph{KDD}}. \bibinfo{pages}{2553--2561}.
\newblock


\bibitem[Izacard and Grave(2021)]%
        {izacard2021distilling}
\bibfield{author}{\bibinfo{person}{Gautier Izacard} {and} \bibinfo{person}{Edouard Grave}.} \bibinfo{year}{2021}\natexlab{}.
\newblock \showarticletitle{Distilling Knowledge from Reader to Retriever for Question Answering}. In \bibinfo{booktitle}{\emph{ICLR}}.
\newblock


\bibitem[Jegou et~al\mbox{.}(2010)]%
        {jegou2010product}
\bibfield{author}{\bibinfo{person}{Herve Jegou}, \bibinfo{person}{Matthijs Douze}, {and} \bibinfo{person}{Cordelia Schmid}.} \bibinfo{year}{2010}\natexlab{}.
\newblock \showarticletitle{Product quantization for nearest neighbor search}.
\newblock \bibinfo{journal}{\emph{TPAMI}} \bibinfo{volume}{33}, \bibinfo{number}{1} (\bibinfo{year}{2010}), \bibinfo{pages}{117--128}.
\newblock


\bibitem[Jiang et~al\mbox{.}(2023)]%
        {jiang2023scaling}
\bibfield{author}{\bibinfo{person}{Ting Jiang}, \bibinfo{person}{Shaohan Huang}, \bibinfo{person}{Zhongzhi Luan}, \bibinfo{person}{Deqing Wang}, {and} \bibinfo{person}{Fuzhen Zhuang}.} \bibinfo{year}{2023}\natexlab{}.
\newblock \showarticletitle{Scaling sentence embeddings with large language models}.
\newblock \bibinfo{journal}{\emph{arXiv preprint arXiv:2307.16645}} (\bibinfo{year}{2023}).
\newblock


\bibitem[Kaplan et~al\mbox{.}(2020)]%
        {kaplan2020scaling}
\bibfield{author}{\bibinfo{person}{Jared Kaplan}, \bibinfo{person}{Sam McCandlish}, \bibinfo{person}{Tom Henighan}, \bibinfo{person}{Tom~B Brown}, \bibinfo{person}{Benjamin Chess}, \bibinfo{person}{Rewon Child}, \bibinfo{person}{Scott Gray}, \bibinfo{person}{Alec Radford}, \bibinfo{person}{Jeffrey Wu}, {and} \bibinfo{person}{Dario Amodei}.} \bibinfo{year}{2020}\natexlab{}.
\newblock \showarticletitle{Scaling laws for neural language models}.
\newblock \bibinfo{journal}{\emph{arXiv preprint arXiv:2001.08361}} (\bibinfo{year}{2020}).
\newblock


\bibitem[Karpukhin et~al\mbox{.}(2020)]%
        {karpukhin2020dense}
\bibfield{author}{\bibinfo{person}{Vladimir Karpukhin}, \bibinfo{person}{Barlas O{\u{g}}uz}, \bibinfo{person}{Sewon Min}, \bibinfo{person}{Patrick Lewis}, \bibinfo{person}{Ledell Wu}, \bibinfo{person}{Sergey Edunov}, \bibinfo{person}{Danqi Chen}, {and} \bibinfo{person}{Wen~Tau Yih}.} \bibinfo{year}{2020}\natexlab{}.
\newblock \showarticletitle{Dense passage retrieval for open-domain question answering}. In \bibinfo{booktitle}{\emph{EMNLP}}. \bibinfo{pages}{6769--6781}.
\newblock


\bibitem[Kenton and Toutanova(2019)]%
        {kenton2019bert}
\bibfield{author}{\bibinfo{person}{Jacob Devlin Ming-Wei~Chang Kenton} {and} \bibinfo{person}{Lee~Kristina Toutanova}.} \bibinfo{year}{2019}\natexlab{}.
\newblock \showarticletitle{Bert: Pre-training of deep bidirectional transformers for language understanding}. In \bibinfo{booktitle}{\emph{NAACL-HLT}}, Vol.~\bibinfo{volume}{1}. \bibinfo{pages}{2}.
\newblock


\bibitem[Khattab and Zaharia(2020)]%
        {khattab2020colbert}
\bibfield{author}{\bibinfo{person}{Omar Khattab} {and} \bibinfo{person}{Matei Zaharia}.} \bibinfo{year}{2020}\natexlab{}.
\newblock \showarticletitle{Colbert: Efficient and effective passage search via contextualized late interaction over bert}. In \bibinfo{booktitle}{\emph{SIGIR}}. \bibinfo{pages}{39--48}.
\newblock


\bibitem[Kim et~al\mbox{.}({[n.\,d.]})]%
        {kimustad}
\bibfield{author}{\bibinfo{person}{Seungyeon Kim}, \bibinfo{person}{Ankit~Singh Rawat}, \bibinfo{person}{Manzil Zaheer}, \bibinfo{person}{Wittawat Jitkrittum}, \bibinfo{person}{Veeranjaneyulu Sadhanala}, \bibinfo{person}{Sadeep Jayasumana}, \bibinfo{person}{Aditya~Krishna Menon}, \bibinfo{person}{Rob Fergus}, {and} \bibinfo{person}{Sanjiv Kumar}.} \bibinfo{year}{[n.\,d.]}\natexlab{}.
\newblock \showarticletitle{USTAD: Unified Single-model Training Achieving Diverse Scores for Information Retrieval}. In \bibinfo{booktitle}{\emph{ICML}}.
\newblock


\bibitem[Kusupati et~al\mbox{.}(2022)]%
        {kusupati2022matryoshka}
\bibfield{author}{\bibinfo{person}{Aditya Kusupati}, \bibinfo{person}{Gantavya Bhatt}, \bibinfo{person}{Aniket Rege}, \bibinfo{person}{Matthew Wallingford}, \bibinfo{person}{Aditya Sinha}, \bibinfo{person}{Vivek Ramanujan}, \bibinfo{person}{William Howard-Snyder}, \bibinfo{person}{Kaifeng Chen}, \bibinfo{person}{Sham Kakade}, \bibinfo{person}{Prateek Jain}, {et~al\mbox{.}}} \bibinfo{year}{2022}\natexlab{}.
\newblock \showarticletitle{Matryoshka representation learning}.
\newblock \bibinfo{journal}{\emph{NeurIPS}}  \bibinfo{volume}{35} (\bibinfo{year}{2022}), \bibinfo{pages}{30233--30249}.
\newblock


\bibitem[Lee et~al\mbox{.}(2024)]%
        {lee2024nv}
\bibfield{author}{\bibinfo{person}{Chankyu Lee}, \bibinfo{person}{Rajarshi Roy}, \bibinfo{person}{Mengyao Xu}, \bibinfo{person}{Jonathan Raiman}, \bibinfo{person}{Mohammad Shoeybi}, \bibinfo{person}{Bryan Catanzaro}, {and} \bibinfo{person}{Wei Ping}.} \bibinfo{year}{2024}\natexlab{}.
\newblock \showarticletitle{NV-Embed: Improved Techniques for Training LLMs as Generalist Embedding Models}.
\newblock \bibinfo{journal}{\emph{arXiv preprint arXiv:2405.17428}} (\bibinfo{year}{2024}).
\newblock


\bibitem[Li et~al\mbox{.}(2024)]%
        {li2024llama2vec}
\bibfield{author}{\bibinfo{person}{Chaofan Li}, \bibinfo{person}{Zheng Liu}, \bibinfo{person}{Shitao Xiao}, \bibinfo{person}{Yingxia Shao}, {and} \bibinfo{person}{Defu Lian}.} \bibinfo{year}{2024}\natexlab{}.
\newblock \showarticletitle{Llama2Vec: Unsupervised Adaptation of Large Language Models for Dense Retrieval}. In \bibinfo{booktitle}{\emph{ACL}}. \bibinfo{pages}{3490--3500}.
\newblock


\bibitem[Li et~al\mbox{.}(2021)]%
        {li2021embedding}
\bibfield{author}{\bibinfo{person}{Sen Li}, \bibinfo{person}{Fuyu Lv}, \bibinfo{person}{Taiwei Jin}, \bibinfo{person}{Guli Lin}, \bibinfo{person}{Keping Yang}, \bibinfo{person}{Xiaoyi Zeng}, \bibinfo{person}{Xiao-Ming Wu}, {and} \bibinfo{person}{Qianli Ma}.} \bibinfo{year}{2021}\natexlab{}.
\newblock \showarticletitle{Embedding-based product retrieval in taobao search}. In \bibinfo{booktitle}{\emph{KDD}}. \bibinfo{pages}{3181--3189}.
\newblock


\bibitem[Liao et~al\mbox{.}(2024)]%
        {liao2024d2llm}
\bibfield{author}{\bibinfo{person}{Zihan Liao}, \bibinfo{person}{Hang Yu}, \bibinfo{person}{Jianguo Li}, \bibinfo{person}{Jun Wang}, {and} \bibinfo{person}{Wei Zhang}.} \bibinfo{year}{2024}\natexlab{}.
\newblock \showarticletitle{D2LLM: Decomposed and Distilled Large Language Models for Semantic Search}. In \bibinfo{booktitle}{\emph{ACL}}. \bibinfo{pages}{14798--14814}.
\newblock


\bibitem[Lin et~al\mbox{.}(2021)]%
        {lin2021batch}
\bibfield{author}{\bibinfo{person}{Sheng-Chieh Lin}, \bibinfo{person}{Jheng-Hong Yang}, {and} \bibinfo{person}{Jimmy Lin}.} \bibinfo{year}{2021}\natexlab{}.
\newblock \showarticletitle{In-batch negatives for knowledge distillation with tightly-coupled teachers for dense retrieval}. In \bibinfo{booktitle}{\emph{Proceedings of the 6th Workshop on Representation Learning for NLP (RepL4NLP-2021)}}. \bibinfo{pages}{163--173}.
\newblock


\bibitem[Liu et~al\mbox{.}(2021a)]%
        {liu2021pre}
\bibfield{author}{\bibinfo{person}{Yiding Liu}, \bibinfo{person}{Weixue Lu}, \bibinfo{person}{Suqi Cheng}, \bibinfo{person}{Daiting Shi}, \bibinfo{person}{Shuaiqiang Wang}, \bibinfo{person}{Zhicong Cheng}, {and} \bibinfo{person}{Dawei Yin}.} \bibinfo{year}{2021}\natexlab{a}.
\newblock \showarticletitle{Pre-trained language model for web-scale retrieval in baidu search}. In \bibinfo{booktitle}{\emph{KDD}}. \bibinfo{pages}{3365--3375}.
\newblock


\bibitem[Liu et~al\mbox{.}(2021b)]%
        {liu2021que2search}
\bibfield{author}{\bibinfo{person}{Yiqun Liu}, \bibinfo{person}{Kaushik Rangadurai}, \bibinfo{person}{Yunzhong He}, \bibinfo{person}{Siddarth Malreddy}, \bibinfo{person}{Xunlong Gui}, \bibinfo{person}{Xiaoyi Liu}, {and} \bibinfo{person}{Fedor Borisyuk}.} \bibinfo{year}{2021}\natexlab{b}.
\newblock \showarticletitle{Que2search: fast and accurate query and document understanding for search at facebook}. In \bibinfo{booktitle}{\emph{KDD}}. \bibinfo{pages}{3376--3384}.
\newblock


\bibitem[Lu et~al\mbox{.}(2020)]%
        {lu2020twinbert}
\bibfield{author}{\bibinfo{person}{Wenhao Lu}, \bibinfo{person}{Jian Jiao}, {and} \bibinfo{person}{Ruofei Zhang}.} \bibinfo{year}{2020}\natexlab{}.
\newblock \showarticletitle{Twinbert: Distilling knowledge to twin-structured compressed bert models for large-scale retrieval}. In \bibinfo{booktitle}{\emph{CIKM}}. \bibinfo{pages}{2645--2652}.
\newblock


\bibitem[Lu et~al\mbox{.}(2024)]%
        {lu2024knowledge}
\bibfield{author}{\bibinfo{person}{Zepu Lu}, \bibinfo{person}{Jin Chen}, \bibinfo{person}{Defu Lian}, \bibinfo{person}{Zaixi Zhang}, \bibinfo{person}{Yong Ge}, {and} \bibinfo{person}{Enhong Chen}.} \bibinfo{year}{2024}\natexlab{}.
\newblock \showarticletitle{Knowledge distillation for high dimensional search index}.
\newblock \bibinfo{journal}{\emph{NeurIPS}}  \bibinfo{volume}{36} (\bibinfo{year}{2024}).
\newblock


\bibitem[Lyu et~al\mbox{.}(2024a)]%
        {lyu2024crud}
\bibfield{author}{\bibinfo{person}{Yuanjie Lyu}, \bibinfo{person}{Zhiyu Li}, \bibinfo{person}{Simin Niu}, \bibinfo{person}{Feiyu Xiong}, \bibinfo{person}{Bo Tang}, \bibinfo{person}{Wenjin Wang}, \bibinfo{person}{Hao Wu}, \bibinfo{person}{Huanyong Liu}, \bibinfo{person}{Tong Xu}, {and} \bibinfo{person}{Enhong Chen}.} \bibinfo{year}{2024}\natexlab{a}.
\newblock \showarticletitle{Crud-rag: A comprehensive chinese benchmark for retrieval-augmented generation of large language models}.
\newblock \bibinfo{journal}{\emph{arXiv preprint arXiv:2401.17043}} (\bibinfo{year}{2024}).
\newblock


\bibitem[Lyu et~al\mbox{.}(2024b)]%
        {lyu2024retrieve}
\bibfield{author}{\bibinfo{person}{Yuanjie Lyu}, \bibinfo{person}{Zihan Niu}, \bibinfo{person}{Zheyong Xie}, \bibinfo{person}{Chao Zhang}, \bibinfo{person}{Tong Xu}, \bibinfo{person}{Yang Wang}, {and} \bibinfo{person}{Enhong Chen}.} \bibinfo{year}{2024}\natexlab{b}.
\newblock \showarticletitle{Retrieve-Plan-Generation: An Iterative Planning and Answering Framework for Knowledge-Intensive LLM Generation}.
\newblock \bibinfo{journal}{\emph{arXiv preprint arXiv:2406.14979}} (\bibinfo{year}{2024}).
\newblock


\bibitem[Ma et~al\mbox{.}(2024a)]%
        {ma2024task}
\bibfield{author}{\bibinfo{person}{Guangyuan Ma}, \bibinfo{person}{Yongliang Ma}, \bibinfo{person}{Xing Wu}, \bibinfo{person}{Zhenpeng Su}, \bibinfo{person}{Ming Zhou}, {and} \bibinfo{person}{Songlin Hu}.} \bibinfo{year}{2024}\natexlab{a}.
\newblock \showarticletitle{Task-level Distributionally Robust Optimization for Large Language Model-based Dense Retrieval}.
\newblock \bibinfo{journal}{\emph{arXiv preprint arXiv:2408.10613}} (\bibinfo{year}{2024}).
\newblock


\bibitem[Ma et~al\mbox{.}(2024b)]%
        {ma2024fine}
\bibfield{author}{\bibinfo{person}{Xueguang Ma}, \bibinfo{person}{Liang Wang}, \bibinfo{person}{Nan Yang}, \bibinfo{person}{Furu Wei}, {and} \bibinfo{person}{Jimmy Lin}.} \bibinfo{year}{2024}\natexlab{b}.
\newblock \showarticletitle{Fine-tuning llama for multi-stage text retrieval}. In \bibinfo{booktitle}{\emph{SIGIR}}. \bibinfo{pages}{2421--2425}.
\newblock


\bibitem[MacQueen(1967)]%
        {macqueen1967some}
\bibfield{author}{\bibinfo{person}{J MacQueen}.} \bibinfo{year}{1967}\natexlab{}.
\newblock \showarticletitle{Some methods for classification and analysis of multivariate observations}. In \bibinfo{booktitle}{\emph{Proceedings of 5-th Berkeley Symposium on Mathematical Statistics and Probability/University of California Press}}.
\newblock


\bibitem[Magnani et~al\mbox{.}(2022)]%
        {magnani2022semantic}
\bibfield{author}{\bibinfo{person}{Alessandro Magnani}, \bibinfo{person}{Feng Liu}, \bibinfo{person}{Suthee Chaidaroon}, \bibinfo{person}{Sachin Yadav}, \bibinfo{person}{Praveen Reddy~Suram}, \bibinfo{person}{Ajit Puthenputhussery}, \bibinfo{person}{Sijie Chen}, \bibinfo{person}{Min Xie}, \bibinfo{person}{Anirudh Kashi}, \bibinfo{person}{Tony Lee}, {et~al\mbox{.}}} \bibinfo{year}{2022}\natexlab{}.
\newblock \showarticletitle{Semantic retrieval at walmart}. In \bibinfo{booktitle}{\emph{KDD}}. \bibinfo{pages}{3495--3503}.
\newblock


\bibitem[Ni et~al\mbox{.}(2022)]%
        {ni2022large}
\bibfield{author}{\bibinfo{person}{Jianmo Ni}, \bibinfo{person}{Chen Qu}, \bibinfo{person}{Jing Lu}, \bibinfo{person}{Zhuyun Dai}, \bibinfo{person}{Gustavo~Hernandez Abrego}, \bibinfo{person}{Ji Ma}, \bibinfo{person}{Vincent Zhao}, \bibinfo{person}{Yi Luan}, \bibinfo{person}{Keith Hall}, \bibinfo{person}{Ming-Wei Chang}, {et~al\mbox{.}}} \bibinfo{year}{2022}\natexlab{}.
\newblock \showarticletitle{Large Dual Encoders Are Generalizable Retrievers}. In \bibinfo{booktitle}{\emph{EMNLP}}. \bibinfo{pages}{9844--9855}.
\newblock


\bibitem[Peng et~al\mbox{.}(2024)]%
        {peng2024large}
\bibfield{author}{\bibinfo{person}{Wenjun Peng}, \bibinfo{person}{Guiyang Li}, \bibinfo{person}{Yue Jiang}, \bibinfo{person}{Zilong Wang}, \bibinfo{person}{Dan Ou}, \bibinfo{person}{Xiaoyi Zeng}, \bibinfo{person}{Derong Xu}, \bibinfo{person}{Tong Xu}, {and} \bibinfo{person}{Enhong Chen}.} \bibinfo{year}{2024}\natexlab{}.
\newblock \showarticletitle{Large language model based long-tail query rewriting in taobao search}. In \bibinfo{booktitle}{\emph{WWW}}. \bibinfo{pages}{20--28}.
\newblock


\bibitem[Pfau et~al\mbox{.}(2024)]%
        {pfau2024let}
\bibfield{author}{\bibinfo{person}{Jacob Pfau}, \bibinfo{person}{William Merrill}, {and} \bibinfo{person}{Samuel~R Bowman}.} \bibinfo{year}{2024}\natexlab{}.
\newblock \showarticletitle{Let's Think Dot by Dot: Hidden Computation in Transformer Language Models}.
\newblock \bibinfo{journal}{\emph{arXiv preprint arXiv:2404.15758}} (\bibinfo{year}{2024}).
\newblock


\bibitem[Qu et~al\mbox{.}(2021)]%
        {qu2021rocketqa}
\bibfield{author}{\bibinfo{person}{Yingqi Qu}, \bibinfo{person}{Yuchen Ding}, \bibinfo{person}{Jing Liu}, \bibinfo{person}{Kai Liu}, \bibinfo{person}{Ruiyang Ren}, \bibinfo{person}{Wayne~Xin Zhao}, \bibinfo{person}{Daxiang Dong}, \bibinfo{person}{Hua Wu}, {and} \bibinfo{person}{Haifeng Wang}.} \bibinfo{year}{2021}\natexlab{}.
\newblock \showarticletitle{RocketQA: An Optimized Training Approach to Dense Passage Retrieval for Open-Domain Question Answering}. In \bibinfo{booktitle}{\emph{NAACL}}. \bibinfo{pages}{5835--5847}.
\newblock


\bibitem[Rasley et~al\mbox{.}(2020)]%
        {rasley2020deepspeed}
\bibfield{author}{\bibinfo{person}{Jeff Rasley}, \bibinfo{person}{Samyam Rajbhandari}, \bibinfo{person}{Olatunji Ruwase}, {and} \bibinfo{person}{Yuxiong He}.} \bibinfo{year}{2020}\natexlab{}.
\newblock \showarticletitle{Deepspeed: System optimizations enable training deep learning models with over 100 billion parameters}. In \bibinfo{booktitle}{\emph{KDD}}. \bibinfo{pages}{3505--3506}.
\newblock


\bibitem[Ren et~al\mbox{.}(2021b)]%
        {ren2021zero}
\bibfield{author}{\bibinfo{person}{Jie Ren}, \bibinfo{person}{Samyam Rajbhandari}, \bibinfo{person}{Reza~Yazdani Aminabadi}, \bibinfo{person}{Olatunji Ruwase}, \bibinfo{person}{Shuangyan Yang}, \bibinfo{person}{Minjia Zhang}, \bibinfo{person}{Dong Li}, {and} \bibinfo{person}{Yuxiong He}.} \bibinfo{year}{2021}\natexlab{b}.
\newblock \showarticletitle{$\{$Zero-offload$\}$: Democratizing $\{$billion-scale$\}$ model training}. In \bibinfo{booktitle}{\emph{USENIX Annual Technical Conference}}. \bibinfo{pages}{551--564}.
\newblock


\bibitem[Ren et~al\mbox{.}(2021a)]%
        {ren2021rocketqav2}
\bibfield{author}{\bibinfo{person}{Ruiyang Ren}, \bibinfo{person}{Yingqi Qu}, \bibinfo{person}{Jing Liu}, \bibinfo{person}{Wayne~Xin Zhao}, \bibinfo{person}{Qiaoqiao She}, \bibinfo{person}{Hua Wu}, \bibinfo{person}{Haifeng Wang}, {and} \bibinfo{person}{Ji-Rong Wen}.} \bibinfo{year}{2021}\natexlab{a}.
\newblock \showarticletitle{RocketQAv2: A Joint Training Method for Dense Passage Retrieval and Passage Re-ranking}. In \bibinfo{booktitle}{\emph{EMNLP}}. \bibinfo{pages}{2825--2835}.
\newblock


\bibitem[Tang and Yang(2024)]%
        {tang2024pooling}
\bibfield{author}{\bibinfo{person}{Yixuan Tang} {and} \bibinfo{person}{Yi Yang}.} \bibinfo{year}{2024}\natexlab{}.
\newblock \showarticletitle{Pooling And Attention: What Are Effective Designs For LLm-Based Embedding Models?}
\newblock \bibinfo{journal}{\emph{arXiv preprint arXiv:2409.02727}} (\bibinfo{year}{2024}).
\newblock


\bibitem[Team(2024)]%
        {qwen2.5}
\bibfield{author}{\bibinfo{person}{Qwen Team}.} \bibinfo{year}{2024}\natexlab{}.
\newblock \bibinfo{title}{Qwen2.5: A Party of Foundation Models}.
\newblock
\newblock
\urldef\tempurl%
\url{https://qwenlm.github.io/blog/qwen2.5/}
\showURL{%
\tempurl}


\bibitem[Wang et~al\mbox{.}(2023)]%
        {wang2023improving}
\bibfield{author}{\bibinfo{person}{Liang Wang}, \bibinfo{person}{Nan Yang}, \bibinfo{person}{Xiaolong Huang}, \bibinfo{person}{Linjun Yang}, \bibinfo{person}{Rangan Majumder}, {and} \bibinfo{person}{Furu Wei}.} \bibinfo{year}{2023}\natexlab{}.
\newblock \showarticletitle{Improving text embeddings with large language models}.
\newblock \bibinfo{journal}{\emph{arXiv preprint arXiv:2401.00368}} (\bibinfo{year}{2023}).
\newblock


\bibitem[Wang et~al\mbox{.}({[n.\,d.]})]%
        {wang2023query2doc}
\bibfield{author}{\bibinfo{person}{Liang Wang}, \bibinfo{person}{Nan Yang}, {and} \bibinfo{person}{Furu Wei}.} \bibinfo{year}{[n.\,d.]}\natexlab{}.
\newblock \showarticletitle{Query2doc: Query Expansion with Large Language Models}. In \bibinfo{booktitle}{\emph{EMNLP}}.
\newblock


\bibitem[Wang and Hong(2023)]%
        {wang2023query}
\bibfield{author}{\bibinfo{person}{Yuxuan Wang} {and} \bibinfo{person}{Lyu Hong}.} \bibinfo{year}{2023}\natexlab{}.
\newblock \showarticletitle{Query Encoder Distillation via Embedding Alignment is a Strong Baseline Method to Boost Dense Retriever Online Efficiency}. In \bibinfo{booktitle}{\emph{Proceedings of The Fourth Workshop on Simple and Efficient Natural Language Processing (SustaiNLP)}}. \bibinfo{pages}{290--298}.
\newblock


\bibitem[Xiong et~al\mbox{.}({[n.\,d.]})]%
        {xiongapproximate}
\bibfield{author}{\bibinfo{person}{Lee Xiong}, \bibinfo{person}{Chenyan Xiong}, \bibinfo{person}{Ye Li}, \bibinfo{person}{Kwok-Fung Tang}, \bibinfo{person}{Jialin Liu}, \bibinfo{person}{Paul~N Bennett}, \bibinfo{person}{Junaid Ahmed}, {and} \bibinfo{person}{Arnold Overwijk}.} \bibinfo{year}{[n.\,d.]}\natexlab{}.
\newblock \showarticletitle{Approximate Nearest Neighbor Negative Contrastive Learning for Dense Text Retrieval}. In \bibinfo{booktitle}{\emph{ICLR}}.
\newblock


\bibitem[Xu et~al\mbox{.}(2023)]%
        {xu2023large}
\bibfield{author}{\bibinfo{person}{Derong Xu}, \bibinfo{person}{Wei Chen}, \bibinfo{person}{Wenjun Peng}, \bibinfo{person}{Chao Zhang}, \bibinfo{person}{Tong Xu}, \bibinfo{person}{Xiangyu Zhao}, \bibinfo{person}{Xian Wu}, \bibinfo{person}{Yefeng Zheng}, {and} \bibinfo{person}{Enhong Chen}.} \bibinfo{year}{2023}\natexlab{}.
\newblock \showarticletitle{Large language models for generative information extraction: A survey}.
\newblock \bibinfo{journal}{\emph{arXiv preprint arXiv:2312.17617}} (\bibinfo{year}{2023}).
\newblock


\bibitem[Xu et~al\mbox{.}(2022)]%
        {xu2022negative}
\bibfield{author}{\bibinfo{person}{Lanling Xu}, \bibinfo{person}{Jianxun Lian}, \bibinfo{person}{Wayne~Xin Zhao}, \bibinfo{person}{Ming Gong}, \bibinfo{person}{Linjun Shou}, \bibinfo{person}{Daxin Jiang}, \bibinfo{person}{Xing Xie}, {and} \bibinfo{person}{Ji-Rong Wen}.} \bibinfo{year}{2022}\natexlab{}.
\newblock \showarticletitle{Negative sampling for contrastive representation learning: A review}.
\newblock \bibinfo{journal}{\emph{arXiv preprint arXiv:2206.00212}} (\bibinfo{year}{2022}).
\newblock


\bibitem[Zhan et~al\mbox{.}(2021)]%
        {zhan2021optimizing}
\bibfield{author}{\bibinfo{person}{Jingtao Zhan}, \bibinfo{person}{Jiaxin Mao}, \bibinfo{person}{Yiqun Liu}, \bibinfo{person}{Jiafeng Guo}, \bibinfo{person}{Min Zhang}, {and} \bibinfo{person}{Shaoping Ma}.} \bibinfo{year}{2021}\natexlab{}.
\newblock \showarticletitle{Optimizing dense retrieval model training with hard negatives}. In \bibinfo{booktitle}{\emph{SIGIR}}. \bibinfo{pages}{1503--1512}.
\newblock


\bibitem[Zhang et~al\mbox{.}(2024a)]%
        {zhang2024simple}
\bibfield{author}{\bibinfo{person}{Bowen Zhang}, \bibinfo{person}{Kehua Chang}, {and} \bibinfo{person}{Chunping Li}.} \bibinfo{year}{2024}\natexlab{a}.
\newblock \showarticletitle{Simple techniques for enhancing sentence embeddings in generative language models}. In \bibinfo{booktitle}{\emph{International Conference on Intelligent Computing}}. Springer, \bibinfo{pages}{52--64}.
\newblock


\bibitem[Zhang et~al\mbox{.}(2024b)]%
        {zhang2024notellm}
\bibfield{author}{\bibinfo{person}{Chao Zhang}, \bibinfo{person}{Shiwei Wu}, \bibinfo{person}{Haoxin Zhang}, \bibinfo{person}{Tong Xu}, \bibinfo{person}{Yan Gao}, \bibinfo{person}{Yao Hu}, {and} \bibinfo{person}{Enhong Chen}.} \bibinfo{year}{2024}\natexlab{b}.
\newblock \showarticletitle{NoteLLM: A Retrievable Large Language Model for Note Recommendation}. In \bibinfo{booktitle}{\emph{WWW}}. \bibinfo{pages}{170--179}.
\newblock


\bibitem[Zhang et~al\mbox{.}(2024c)]%
        {zhang2024notellm2}
\bibfield{author}{\bibinfo{person}{Chao Zhang}, \bibinfo{person}{Haoxin Zhang}, \bibinfo{person}{Shiwei Wu}, \bibinfo{person}{Di Wu}, \bibinfo{person}{Tong Xu}, \bibinfo{person}{Yan Gao}, \bibinfo{person}{Yao Hu}, {and} \bibinfo{person}{Enhong Chen}.} \bibinfo{year}{2024}\natexlab{c}.
\newblock \showarticletitle{NoteLLM-2: Multimodal Large Representation Models for Recommendation}.
\newblock \bibinfo{journal}{\emph{arXiv preprint arXiv:2405.16789}} (\bibinfo{year}{2024}).
\newblock


\bibitem[Zhang et~al\mbox{.}(2022)]%
        {zhang2022uni}
\bibfield{author}{\bibinfo{person}{Jianjin Zhang}, \bibinfo{person}{Zheng Liu}, \bibinfo{person}{Weihao Han}, \bibinfo{person}{Shitao Xiao}, \bibinfo{person}{Ruicheng Zheng}, \bibinfo{person}{Yingxia Shao}, \bibinfo{person}{Hao Sun}, \bibinfo{person}{Hanqing Zhu}, \bibinfo{person}{Premkumar Srinivasan}, \bibinfo{person}{Weiwei Deng}, {et~al\mbox{.}}} \bibinfo{year}{2022}\natexlab{}.
\newblock \showarticletitle{Uni-retriever: Towards learning the unified embedding based retriever in bing sponsored search}. In \bibinfo{booktitle}{\emph{KDD}}. \bibinfo{pages}{4493--4501}.
\newblock


\bibitem[Zhao et~al\mbox{.}(2024)]%
        {zhao2024dense}
\bibfield{author}{\bibinfo{person}{Wayne~Xin Zhao}, \bibinfo{person}{Jing Liu}, \bibinfo{person}{Ruiyang Ren}, {and} \bibinfo{person}{Ji-Rong Wen}.} \bibinfo{year}{2024}\natexlab{}.
\newblock \showarticletitle{Dense text retrieval based on pretrained language models: A survey}.
\newblock \bibinfo{journal}{\emph{ACM TOIS}} \bibinfo{volume}{42}, \bibinfo{number}{4} (\bibinfo{year}{2024}), \bibinfo{pages}{1--60}.
\newblock


\bibitem[Zhao et~al\mbox{.}(2023)]%
        {zhao2023survey}
\bibfield{author}{\bibinfo{person}{Wayne~Xin Zhao}, \bibinfo{person}{Kun Zhou}, \bibinfo{person}{Junyi Li}, \bibinfo{person}{Tianyi Tang}, \bibinfo{person}{Xiaolei Wang}, \bibinfo{person}{Yupeng Hou}, \bibinfo{person}{Yingqian Min}, \bibinfo{person}{Beichen Zhang}, \bibinfo{person}{Junjie Zhang}, \bibinfo{person}{Zican Dong}, {et~al\mbox{.}}} \bibinfo{year}{2023}\natexlab{}.
\newblock \showarticletitle{A survey of large language models}.
\newblock \bibinfo{journal}{\emph{arXiv preprint arXiv:2303.18223}} (\bibinfo{year}{2023}).
\newblock


\bibitem[Zhu et~al\mbox{.}(2023)]%
        {zhu2023large}
\bibfield{author}{\bibinfo{person}{Yutao Zhu}, \bibinfo{person}{Huaying Yuan}, \bibinfo{person}{Shuting Wang}, \bibinfo{person}{Jiongnan Liu}, \bibinfo{person}{Wenhan Liu}, \bibinfo{person}{Chenlong Deng}, \bibinfo{person}{Haonan Chen}, \bibinfo{person}{Zhicheng Dou}, {and} \bibinfo{person}{Ji-Rong Wen}.} \bibinfo{year}{2023}\natexlab{}.
\newblock \showarticletitle{Large language models for information retrieval: A survey}.
\newblock \bibinfo{journal}{\emph{arXiv preprint arXiv:2308.07107}} (\bibinfo{year}{2023}).
\newblock


\end{thebibliography}
